\documentclass[prerint,11pt]{article}
\usepackage{setspace}
\usepackage{appendix}

\usepackage{xfrac}
\usepackage{amssymb}

\usepackage[figuresright]{rotating}

\usepackage{amsmath,amssymb}
\usepackage{graphicx}
\usepackage{dcolumn}
\usepackage{bm}
\usepackage{amscd,amsthm}
\usepackage{ifthen}
\usepackage{mathtools}

\usepackage{pgf,tikz}
\usepackage{pgfplots}
\usetikzlibrary{calc,shadows}
\usetikzlibrary{pgfplots.groupplots}
\usepackage{subfigure}
\usepackage{url}
\usetikzlibrary{pgfplots.groupplots}

\usetikzlibrary{external}
\tikzset{external/force remake}

\usepackage{amsthm}
\newtheorem{lemma}{Lemma}
\newtheorem{definition}{Definition}

\newtheorem{theorem}{Theorem}

\newtheorem{proposition}{Proposition}

\newtheorem*{remark*}{Remark}
 \newcommand{\ve}[1]{\mathbf{#1} }

 \newcommand{\mc}[0]{\mathcal }
 \newcommand{\mb}[0]{\mathbb }
 \newcommand{\mbf}[0]{\mathbf }

\newcommand\norm[2][\Tnorm]{\ensuremath{{\left\Vert #2 \right\Vert}_{#1}}}

\newcommand\Tinnerprod{}
\newcommand{\innerprod}[3][\Tinnerprod]{\ifthenelse{\equal{#1}{}}{\ensuremath{\left<#2,#3\right>}}{\ensuremath{\left<#2,#3\right>_{#1}}}}

\newcommand\vect[1]{\mathbf #1}

\newcommand{\range}{\mc R}
\newcommand{\area}{\mc A}

\newcommand{\safemath}[2]{\newcommand{#1}{\ensuremath{#2}}}
\safemath{\complexset}{\mathbb{C}}

\newcommand\defeq{\coloneqq}

\newcommand{\Ltwo}[0]{\mathrm{L}^2} 

\newcommand\Tex{}
\newcommand\Ex[2][\Tex]{%
\ifthenelse{\equal{#1}{}}{{\mathbb E}[#2]}{\ensuremath{\underset{#1}{\mathbb E}\left[ #2\right]}}}

\DeclareMathOperator{\supp}{supp}
\DeclareMathOperator{\sinc}{sinc}
\DeclareMathOperator{\spanof}{span}
\DeclareMathOperator{\rnk}{rank}

\newcommand{\pinv}[1]{  {#1}^{ \dagger } } 

\newcommand{\conj}[1]{ {#1}^* } 
\newcommand{\herm}[1]{{#1}^H} 
\newcommand{\transp}[1]{{#1}^T} 
\newcommand{\rank}[1]{ \rnk\! \left({#1}\right) } 
\newcommand{\mtx}[1]{\mathbf #1}

\newcommand\EX[2][\Tex]{
\ifthenelse{\equal{#1}{}}{{\mathbb E}[#2]}{\ensuremath{\underset{#1}{\mathbb E}\left[ #2\right]}}}

\newcommand{\coh}[1]{%
  \ifthenelse{\isempty{#1}}%
    {\mu }
    {\mu (#1)}
}

\DeclareMathOperator{\spark}{spark}
\newcommand{\Un}[0]{ U }  
\newcommand{\SSF}[0]{ \Delta  }  
\newcommand{\opfrom}[0]{ X } 
\newcommand{\opto}[0]{  Y } 
\newcommand{\sfunc}{s_{\hspace{-0.02cm}H}} 
\newcommand{\csamp }[0]{T} 

\newcommand{\opclassgen}[0]{\ensuremath{ \mathcal Q}} 
\newcommand{\opclass}[0]{\ensuremath{ \mathcal H}}  

    \newcommand{\opclassr}[0]{\ensuremath{ \mathcal H_{M_\Gamma}  }}    
    \newcommand{\opclassbr}[0]{ \mathcal X (\SSF)  }                    

    \newcommand{\prsig}[0]{x}  

\newcommand{\jpg}{\mc C\mc N}

\renewcommand{\i}{i} 				
\newcommand{\I}{M} 				

\newcommand{\sfuncm}[1]{s_{\hspace{-0.02cm}H_{#1}}}

\newcommand{\sfuncspace}{S}

\usepackage{hyperref}
\hypersetup{
    colorlinks=false,       
    linkcolor=red,          
    citecolor=green,        
    filecolor=magenta,      
    urlcolor=cyan           
}

\begin{document}

\title{Identification of Sparse Linear Operators}

\renewcommand\footnotemark{}

\author{Reinhard Heckel and Helmut B\"olcskei\thanks{Part of this paper was presented at the 2011 IEEE International Symposium on Information Theory (ISIT) \cite{heckel_compressive_2011}.}
 \\[0.5em]
  \multicolumn{1}{p{.7\textwidth}}{\centering Dept. of IT \& EE, ETH Zurich, Switzerland}}

%
%

\maketitle

\renewcommand{\sfrac}[2]{#1/#2}

\begin{abstract}
We consider the problem of identifying a linear deterministic operator from its response to a given probing signal. For a large class of linear operators, we show that stable identifiability is possible if the total support area of the operator's spreading function satisfies $\SSF \leq \sfrac{1}{2}$. This result holds for an arbitrary (possibly fragmented) support region of the spreading function, does not impose limitations on the total extent of the support region, and, most importantly, does not require the support region to be known prior to identification. Furthermore, we prove that stable identifiability of \emph{almost all operators} is possible if $\SSF < 1$. This result is surprising as it says that there is no penalty for not knowing the support region of the spreading function prior to identification. Algorithms that provably recover all operators with $\SSF \leq \sfrac{1}{2}$, and almost all operators with $\SSF < 1$ are presented. 
\end{abstract}


\section{Introduction}
The identification of a deterministic linear operator from the operator's response to a probing signal is an important problem in many fields of engineering. Concrete examples include system identification in control theory and practice, the measurement of dispersive communication channels, and radar imaging. 
It is natural to ask under which conditions (on the operator) identification is possible, in principle, and how one would go about choosing the probing signal and extracting the operator from the corresponding output signal. This paper addresses these questions by considering the (large) class of 
linear operators that can be represented as a continuous weighted superposition of time-frequency shift operators, i.e., the operator's response to the signal $\prsig(t)$ can be written as
\begin{equation}
y(t) = \int_\tau \int_\nu  \sfunc (\tau,\nu) \prsig(t-\tau)  e^{j 2\pi \nu t}  d\nu d\tau
\label{eq:ltvsys}
\end{equation}
where $\sfunc(\tau,\nu)$ denotes the spreading function associated with the operator. 
The representation theorem \cite[Thm. 14.3.5]{groechenig_foundations_2001} states that the action of a large class of continuous (and hence bounded) linear operators can be represented as in \eqref{eq:ltvsys}. 
In the communications literature operators with input-output relation as in \eqref{eq:ltvsys} are referred to as linear time-varying (LTV) channels/systems and $\sfunc(\tau,\nu)$ is the delay-Doppler spreading function \cite{kailath_measurements_1962,bello_measurement_1969,bello_characterization_1963}. 

For the special case of linear time-invariant (LTI) systems, we have $\sfunc(\tau,\nu) = h(\tau)\delta (\nu)$, so that \eqref{eq:ltvsys} reduces to the standard convolution relation
\begin{equation}
y(t) = \int_\tau h(\tau) \prsig (t-\tau) d\tau.
\label{eq:lti_system}
\end{equation}
The question of identifiability of LTI systems is readily answered by noting that the system's response to the Dirac delta function  
is given by the impulse response $h(t)$, which by \eqref{eq:lti_system} fully characterizes the system's input-output relation. LTI systems are therefore always identifiable, provided that the probing signal can have infinite bandwidth and we can observe the output signal over an infinite duration. 

For LTV systems the situation is fundamentally different. Specifically, 
 Kailath's landmark paper  \cite{kailath_measurements_1962} shows that an LTV system with spreading function compactly supported on a rectangle of area $\SSF$ is identifiable if and only if $\SSF \leq 1$. 
This condition can be very restrictive. Measurements of underwater acoustic communication channels, such as those reported in \cite{eggen_underwater_1997} for example, show that the support area of the spreading function can be larger than $1$. The measurements in \cite{eggen_underwater_1997} exhibit, however, an interesting structural property: 
The nonzero components of the spreading function are scattered across the $(\tau, \nu)$-plane and the sum of the corresponding support areas, henceforth called ``overall support area'', is smaller than $1$. A similar situation arises in radar astronomy \cite{hagfors_mapping_1991}. 
 Bello \cite{bello_measurement_1969} shows that Kailath's identifiability result continues to hold for arbitrarily fragmented spreading function support regions as long as the corresponding overall support area is smaller than $1$. 
 Kozek and Pfander \cite{kozek_identification_2005}  and Pfander and Walnut \cite{pfander_measurement_2006} found elegant functional-analytical identifiability proofs for setups that are more general than those originally considered in \cite{kailath_measurements_1962} and \cite{bello_measurement_1969}. 
However, the results in \cite{kailath_measurements_1962,bello_measurement_1969,kozek_identification_2005,pfander_measurement_2006} require the support region of $\sfunc(\tau,\nu)$ to be known prior to identification, a condition that is very restrictive and often impossible to realize in practice. In the case of underwater acoustic communication channels, e.g., the support area of $\sfunc(\tau,\nu)$ depends critically on surface motion, water depth, and motion of transmitter and receiver. For wireless channels, knowing the spreading function's support region would amount to knowing the delays and Doppler shifts induced by the scatterers in the propagation medium. 

\paragraph*{Contributions}
We show that an operator with input-output relation \eqref{eq:ltvsys} is identifiable, without prior knowledge  of the operator's spreading function support region and without limitations on its total extent, if and only if the spreading function's total support area satisfies $\SSF  \leq \sfrac{1}{2}$. 
What is more, this factor-of-two penalty---relative to the case where the support region is known prior to identification \cite{kailath_measurements_1962,bello_measurement_1969,kozek_identification_2005,pfander_measurement_2006}---can be eliminated 
if one asks for identifiability of \emph{almost all}\footnote{Here, and in the remainder of the paper, ``almost all'' is to be understood in a measure-theoretic sense meaning that the set of exceptions has measure zero.} operators only. This result is surprising as it says that (for almost all operators) there is no price to be paid for not knowing the spreading function's support region in advance. 
Our findings have strong conceptual parallels to the theory of spectrum-blind sampling of sparse multi-band signals \cite{feng_spectrum-blind_1996,feng_universal_1997,lu_theory_2008,mishali_blind_2009,eldar_compressed_2009}.  

Furthermore, we present algorithms which, in the noiseless case, provably recover all operators with $\SSF \leq \sfrac{1}{2}$, and almost all operators with $\SSF < 1$, without requiring prior knowledge of the spreading function's support region; not even its area $\SSF$ has to be known.  
Specifically, we formulate the recovery problem as a continuous multiple measurement vector (MMV) problem \cite{mishali_reduce_2008}. 
We then show that this problem can be reduced to a finite MMV problem \cite{jie_chen_theoretical_2006}. The reduction approach we present is of independent interest as it unifies a number of reduction approaches available in the literature and presents a simplified treatment. 

In the case of wireless channels or radar systems, the spreading function's support region is sparse and typically contained in a rectangle of area $1$. In the spirit of compressed sensing, where sparse objects are reconstructed by taking fewer measurements than mandated by their ``bandwidth'', we show that in this case sparsity (in the spreading function's support region) can be exploited to identify the system while undersampling the response to the probing signal. In the case of channel identification this allows for a reduction of the identification time, and in radar systems it leads to increased resolution. 

\paragraph*{Relation to previous work}
Recently, Taub\"{o}ck et al. \cite{taubock_compressive_2010} and Bajwa et al. \cite{bajwa_learning_2008,bajwa_identification_2011} 
considered the identification of LTV systems with 
spreading function compactly supported  in a rectangle of area $\Delta \leq 1$. While \cite{bajwa_learning_2008,bajwa_identification_2011} assume that the spreading function consists of a finite number of Dirac components whose delays and Doppler shifts are unknown prior to identification, the methods proposed in \cite{taubock_compressive_2010} do not need this assumption. 
In the present paper, we allow general (i.e., continuous, discrete, or mixed continuous-discrete) spreading functions that can be supported in the entire $(\tau,\nu)$-plane with possibly fragmented support region. 
Herman and Strohmer \cite{herman_high-resolution_2009}, in the context of compressed sensing radar, and Pfander et al. \cite{pfander_identification_2008} considered the problem of identifying finite-dimensional matrices that are sparse in the basis of time-frequency shift matrices. 
This setup can be obtained from ours by discretization of the input-output relation \eqref{eq:ltvsys} through band-limitation of the input signal and time-limitation and sampling of the output signal. 
The signal recovery problem in \cite{taubock_compressive_2010,bajwa_learning_2008,bajwa_identification_2011,herman_high-resolution_2009,pfander_identification_2008} is a standard \emph{single measurement} recovery problem \cite{elad_sparse_2010}. 
As we start from a continuous-time formulation we find that, depending on the resolution induced by the discretization through time/band-limitation, the resulting recovery problem can be an MMV problem. This is relevant as multiple measurements can improve the recovery performance significantly. In fact, it is the MMV nature of the recovery problem that allows identification of \emph{almost all} operators with $\SSF <1$. 

\paragraph*{Organization of the paper}
The remainder of the paper is organized as follows. In Section \ref{sec:prform}, we formally state the problem considered. Section \ref{sec:mainresults} contains our main identifiability results with the corresponding proofs given in Sections \ref{sec:proofmthm} and \ref{sec:almostall}. 
In Sections \ref{sec:reconstruct} and \ref{sec:almostall}, we present identifiability algorithms along with corresponding performance guarantees. In Section \ref{sec:discior}, we consider the identification of systems with sparse spreading function compactly supported within a rectangle of area $1$. Section \ref{sec:numerical_recres} contains numerical results.

\paragraph*{Notation} The superscripts $\conj{}$, $\herm{}$, and $\transp{}$ stand for complex conjugation, Hermitian transposition, and transposition, respectively.  We use lowercase boldface letters to denote (column) vectors, e.g., $\mbf x$, and uppercase boldface letters to designate matrices, e.g., $\mbf X$. The entry in the $k$th row and $l$th column of $\mbf X$ is $[\mbf X]_{k,l}$ and the $k$th entry of $\vect{x}$ is $[\vect{x}]_k$. The Euclidean norm of $\mbf x$ is denoted by $\norm[2]{\mbf x}$, and $\norm[0]{\mbf x}$ stands for the number of non-zero entries in $\vect{x}$. 
The space spanned by the columns of $\mbf X$ is  $\range(\mbf X)$, and the nullspace of $\mtx{X}$ is denoted by $\ker (\mtx{X})$. $\spark(\mtx{X})$ designates the cardinality of the smallest set of linearly dependent columns of $\mtx{X}$. 

$|\Omega|$ stands for the cardinality of the set $\Omega$. For sets $\Omega_1$ and $\Omega_2$, we define set addition as $\Omega_1+\Omega_2= \{\omega\colon \omega= \omega_1 +\omega_2, \, \omega_1 \in \Omega_1, \omega_2 \in \Omega_2\}$. 
For a (multi-variate) function $f(\mbf x)$, $\supp (f)$ denotes its support set. For (multi-variate) functions $f(\mbf x)$ and $g(\mbf x)$, both with domain $\Omega$, we write $\innerprod[]{f}{g} \defeq \int_\Omega f(\mbf x) \conj{g}\!(\mbf x) d \mbf x$ for their inner product and $\norm{f} \defeq \sqrt{ \innerprod[]{f}{f}}$ for the norm of $f$. 
We say that a set of 
functions $\{g_1(\mbf x), ...,g_n(\mbf x)\}$ with domain $\Omega$ is linearly independent if there is no vector $\vect{a} \in \complexset^n, \vect{a} \neq \vect{0}$, such that $\herm{\vect{a}} \vect{g}(\vect{x}) = 0, \; \forall \vect{x} \in \Omega$, where  $\vect{g}(\vect{x}) = \transp{[g_1(\vect{x}), ...,  g_n(\vect{x})]}$. 
We denote the dimension of the span of $\{g_1(\vect{x}), ...,g_n(\vect{x})\}$ as $\dim \spanof \{ g_1(\vect{x}), ...,g_n(\vect{x})  \}$.

The Fourier transform of a function $x(t)$ is defined as $X(f) =\int_t x(t) e^{-j2\pi ft} dt$. The Dirac delta function is denoted by  $\delta(t)$ and $\sinc(t) \defeq \sin(\pi t)/(\pi t)$. $L^2(\mb R)$ stands for the space of complex-valued square-integrable functions. 

The random variable $X \sim \jpg(m,\sigma^2)$ is proper complex Gaussian with mean $m$ and variance $\sigma^2$. Finally, for $x\in \mathbb R$, we let $\lfloor x \rfloor$ be the largest integer not greater than $x$.  

\section{Problem statement \label{sec:prform}}
Given the normed linear spaces $\opfrom$ and $\opto$, we consider linear operators $H\colon \opfrom \to \opto$ that can be represented as a weighted superposition of translation operators $T_\tau$, with $(T_{\tau} \prsig )(t) \defeq \prsig (t-\tau), x\in \opfrom$, and modulation operators $M_{\nu}$, with $(M_\nu \prsig)(t) \defeq e^{j2\pi \nu t} \prsig (t), x\in \opfrom$, according to
\begin{equation}
(H\prsig)(t) \defeq \int_\tau \int_\nu  \sfunc (\tau,\nu) ( M_{\nu}  T_{\tau} \prsig )(t)   d\nu \hspace{0.02cm}  d\tau 
\label{eq:hilbertschmidt_s}
\end{equation}
with $\sfunc \in \sfuncspace$, where $\sfuncspace$ is a normed linear space.  
This is a rather general setup, since according to \cite[Thm. 14.3.5]{groechenig_foundations_2001}, a large class of continuous (and hence bounded) linear operators can be represented as in \eqref{eq:hilbertschmidt_s}. 
For the theory to be mathematically precise, we need to consider suitable triplets of spaces $(\opfrom,\opto,\sfuncspace)$ inducing a space of operators $\mc H = \mc H(\opfrom,\opto,\sfuncspace)$. 
The triplets $(\opfrom, \opto,\sfuncspace)$ have to be chosen to ``match''; specifically, $\sfuncspace$ may have to satisfy certain regularity conditions, depending on the choice of $\opfrom$, for \eqref{eq:hilbertschmidt_s} to be well-defined. 
For example, if $\opfrom=\opto$ is a Hilbert space and $\sfuncspace = L^2(\mb R^2)$, then $\mc H$ is the set of Hilbert-Schmidt operators \cite[p.~331, A.8]{groechenig_foundations_2001}.  
 Since our identifiability proof relies on the use of Dirac delta functions as probing signals, we need to choose $\opfrom$ such that it contains generalized functions. A corresponding valid choice is the following \cite{pfander_measurement_2006}: Let $\sfuncspace$ and $\opfrom$ be Feichtinger's Banach algebra  $S_0(\mb R^2)$ \cite{feichtinger_banach_1998} and its dual $S_0'(\mb R)$, respectively, and $\opto = L^2(\mb R)$. 
The corresponding space $\mc H$ is equipped with the Hilbert-Schmidt norm 
\begin{align}
\norm[\opclass]{H} = \norm[\Ltwo]{\sfunc} = \left( \int_
\tau \int_\nu \left| \sfunc(\tau,\nu)\right|^2 \right)^{1/2}.
\label{eq:defHSnorm}
\end{align}
Another valid triplet is obtained by setting $\opfrom = S_0'(\mb R), \opto = S_0'(\mb R)$, and $\sfuncspace = S'_0(\mb R^2)$, 
see \cite{pfander_operator_2006,pfander_sampling_2010}. In this case both $\opfrom$ and $\sfuncspace$ contain generalized functions, in particular Dirac delta functions; however, the norm of the corresponding space $\opclass$ takes on a more complicated form than \eqref{eq:defHSnorm}. 
The norm the space $\opclass$ is equipped with determines the definition of identifiability. While our results, including the identification algorithms, hold true for the setup $(\opfrom = S_0'(\mb R), \opto = S_0'(\mb R),\sfuncspace = S'_0(\mb R^2))$, to keep the exposition simple, we will work with the triplet $(\opfrom = S'_0(\mb R), Y = L^2(\mb R),\sfuncspace = S_0(\mb R^2))$  and the corresponding norm \eqref{eq:defHSnorm}. 
We refer the interested reader to \cite{kozek_identification_2005,pfander_measurement_2006,pfander_operator_2006,pfander_sampling_2010} for a detailed description of the rigorous mathematical setup required for the choice $(\opfrom = S_0'(\mb R), \opto = S_0'(\mb R),\sfuncspace = S'_0(\mb R^2))$. 

\paragraph*{Restrictions on the spreading function} 
Following \cite{kailath_measurements_1962,bello_measurement_1969,kozek_identification_2005,pfander_measurement_2006} we consider spreading functions with compact support. This assumption is not critical and can be justified, e.g., in wireless and in underwater acoustic communication channels as follows. 
The extent of the spreading function in $\tau~(\nu)$-direction is determined by the maximum delay (Doppler shift) induced by the channel. The maximum Doppler shift will be limited as the velocity of objects in the propagation channel and/or the velocity of transmitter and receiver is limited. While the maximum delay induced by scattering objects in the channel can, in principle, be arbitrarily large, contributions corresponding to large enough delay will be sufficiently small to be treated as additive noise, thanks to path loss \cite{tse_fundamentals_2005}.  While we do not present analytical results for the noisy case, the impact of noise on the performance of our identification algorithms is assessed numerically in Section \ref{sec:numerical_recres}.

Following \cite{kozek_identification_2005,pfander_measurement_2006} we, moreover, restrict ourselves to spreading functions with support regions of the form 
\begin{equation}
M_\Gamma  \defeq \!\! \bigcup_{(k,m)\, \in \, \Gamma } \left( \Un + \left(k \csamp ,  \frac{m}{\csamp L} \right) \right)  \subseteq [0, \tau_{\max}) \times [0, \nu_{\max})
\label{eq:suppgridcont}
\end{equation}
where $\Un \defeq \left[ 0, \csamp  \right) \times \left[0, 1/ (\csamp L)  \right)$ is a ``cell'' in the $(\tau,\nu)$-plane and $L \in \mathbb N^+$ 
and $\csamp \in \mb R^+$ are parameters whose role will become clear shortly. The set of ``active cells'' is specified by $\Gamma  \subseteq \Sigma \defeq \{(0,0), (0,1),...,(L-1,L-1)\}$. Since $\tau_{\max} = \csamp L$ and $\nu_{\max} = 1/\csamp$, it follows that, choosing $L$ and $T$ accordingly, $\tau_{\max}$ and $\nu_{\max}$ can be arbitrarily large; the spreading function can hence be supported on an arbitrarily large, but finite, region. We denote the area of  $M_\Gamma$ as $\area(M_\Gamma)$ and note that $\area(M_\Gamma) = |\Gamma| \area(\Un) = |\Gamma|/L$.  

A general, possibly fragmented, support region of the spreading function can  be approximated arbitrarily well by covering it with rectangles $\Un$ (see Figure \ref{fig:supex}), as in \eqref{eq:suppgridcont}, with $\csamp$ and $L$ chosen suitably. Note that $L$ determines how fine this approximation is, since $\area(U) = 1/L$, while $T$ controls the ratio of width to height of $\Un$. 
Characterizing the identifiability of operators whose spreading function support region is not compact, but has finite area, e.g., $\supp(\sfunc) \subseteq \{(\tau,\nu)\colon 0\leq \nu \leq \infty, 0 \leq \tau \leq e^{-\nu}  \}$, is an open problem \cite{pfander_private_2011}.  
\begin{figure}[h!]
\begin{center}

\includegraphics{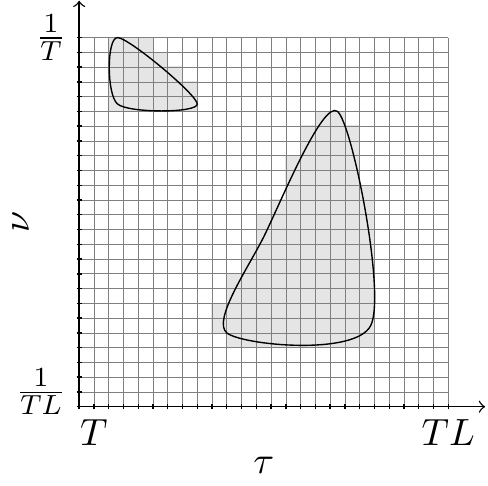}
\end{center}
\caption{
Approximation of a general spreading function support region. }
\label{fig:supex}
\end{figure}
\subsection{Identifiability \label{sec:identf}}

Let us next define  the notion of identifiability
of a set of operators $\opclassgen \subseteq \opclass$. 
The set $\opclassgen$ is said to be identifiable, if there exists a probing signal $\prsig \in \opfrom$ such that for each operator $H\in \opclassgen$, the action of the operator on the probing signal, $H\prsig$, uniquely determines $H$. More formally, we say that  $\opclassgen$ is identifiable if there exists an $\prsig \in \opfrom$ such that 
\begin{equation}
H_1 \prsig = H_2 \prsig \implies H_1=H_2, \quad \forall \hspace{0.1cm} H_1,H_2 \in \opclassgen.
\label{cond:onetoone}
\end{equation}
Identifiability is hence equivalent to invertibility of the mapping
\begin{equation}
T_\prsig\colon \opclassgen \to \opto\colon H \mapsto H \prsig
\label{eq:mapping} 
\end{equation}   
induced by the probing signal $\prsig$. 
In practice, invertibility alone is not sufficient as we want to recover $H$ from $H \prsig$ in a numerically stable fashion, i.e., we want small errors in $H \prsig$ to result in small errors in the identified operator. 
This requirement implies that the inverse of the mapping \eqref{eq:mapping} must be continuous (and hence bounded), which finally motivates the following definition of (stable) identifiability, used in the remainder of the paper.  
\begin{definition}
\label{def:defstableidentification}
We say that $\prsig$ identifies $\opclassgen$ if there exist constants $0 < \alpha \leq \beta < \infty$ such that for all pairs $H_1,H_2 \in \opclassgen$,
\begin{equation}
\alpha \norm[\opclass]{H_1 -H_2} \leq \norm[]{ H_1\prsig - H_2 \prsig  } \leq \beta \norm[\opclass]{ H_1-H_2 }.
\label{eq:cond_stable}
\end{equation}
Furthermore, we say that $\opclassgen$ is identifiable, if there exists an $\prsig \in \opfrom$ such that $\prsig$ identifies $\opclassgen$. 
\end{definition}

In \cite{kozek_identification_2005} the identification of operators of the form \eqref{eq:hilbertschmidt_s} under the assumption of the spreading function support region being known prior to identification is considered. 
The set $\opclassgen$ in \cite{kozek_identification_2005} therefore consists of operators with spreading function supported on a given region $M_\Gamma$, (i.e., $\supp(\sfunc) \subseteq M_{\Gamma}$ for all $H \in \opclassgen$),  which renders $\opclassgen$ a linear subspace of $\opclass$ so that $(H_1-H_2) \in \opclassgen$, for all $H_1,H_2 \in \opclassgen$. Hence Definition \ref{def:defstableidentification} above is equivalent to the following: $\prsig$ identifies $\opclassgen$ if there exist constants $0 < \alpha \leq \beta < \infty$ such that for all $H\in \opclassgen$, $\alpha \norm[\opclass]{H} \leq \norm[]{ H\prsig  } \leq \beta \norm[\opclass]{ H}$, which is the identifiability condition put forward in \cite{kozek_identification_2005}. Not knowing the spreading function's support region prior to identification will require consideration of sets $\opclassgen$ that are not linear subspaces of $\opclass$, which makes the slightly more general Definition \ref{def:defstableidentification} necessary. As detailed in Appendix \ref{app:binv}, the lower bound in \eqref{eq:cond_stable} guarantees that the inverse of $T_\prsig$ in \eqref{eq:mapping} exists and is bounded and hence continuous, as desired. 
The ratio $\beta/\alpha$ quantifies the noise sensitivity of the identification process. Specifically, suppose that $\prsig$ identifies $\opclassgen$, but the measurement $H_1\prsig$, $H_1 \in \opclassgen$, is corrupted by additive noise. Concretely, assume that instead of $H_1\prsig$, we observe $H_1\prsig + w$, where $w\in \opto$ is bounded, i.e., $\norm[]{w}<\infty$. 
Now assume that $H_2 \in \opclassgen$ is consistent with the noisy observation $H_1 \prsig + w$, i.e., $H_1 \prsig +w=H_2\prsig$. We would like the error in the identified operator, i.e., $\norm[\opclass]{H_1-H_2}$, to be proportional to $\norm{w}$. The lower bound in \eqref{eq:cond_stable} guarantees that this is, indeed, the case as 
\begin{equation}
\norm[\opclass]{H_1 -H_2} \leq \frac{1}{\alpha}  \norm[]{ H_1\prsig - H_2 \prsig } = \frac{1}{\alpha}  \norm[]{ w }.   
\label{eq:erpropnoise}
\end{equation}
Since $\alpha \leq \beta$, it follows from \eqref{eq:erpropnoise} that $\beta/\alpha = 1$ is optimal in terms of noise sensitivity. 
We can also conclude from \eqref{eq:erpropnoise} that larger $\alpha$ leads to smaller noise sensitivity. Increased $\alpha$, however, simply amounts to increased power of the probing signal. This can be seen as follows. 
 Suppose we found an $\prsig_1$ that identifies $\opclassgen$ with constants $\alpha,\beta$ in \eqref{eq:cond_stable}. Then, $c \, \prsig_1$ with $c\in \mathbb R$ identifies $\opclassgen$ with constants $|c| \alpha, |c|\beta$. Choosing $|c|$ large will therefore lead to small noise sensitivity. 

\section{Main results \label{sec:mainresults}}
Before stating our main results, 
we define the set of operators with spreading function supported on a given region $M_\Gamma$ (with $M_\Gamma$ as defined in \eqref{eq:suppgridcont}):
\begin{equation}
\opclassr  \defeq \{ H \in \mc H\colon \supp(\sfunc) \subseteq M_{\Gamma} \}.
\end{equation}
Kailath \cite{kailath_measurements_1962} and Kozek and Pfander \cite{kozek_identification_2005} considered the case where $M_\Gamma$ is a (single) rectangle, and Bello \cite{bello_measurement_1969} and Pfander and Walnut \cite{pfander_measurement_2006} analyzed  the case where $M_\Gamma$ is allowed to be fragmented and spread over the $(\tau,\nu)$-plane. In both cases the support region $M_\Gamma$ is assumed to be known prior to identification. 
We start by recalling the key result in \cite{pfander_measurement_2006}, which subsumes the results in \cite{kailath_measurements_1962,bello_measurement_1969}, and \cite{kozek_identification_2005}.
\begin{theorem}[\cite{pfander_measurement_2006}]
Let $M_\Gamma$ be given. 
The set of operators $\opclassr$ is identifiable if and only if $\area(M_\Gamma) \leq 1$.
\label{pr:necessity}
\label{pro:knownsupport}
\end{theorem}

As mentioned earlier, knowing the support region $M_\Gamma$ prior to identification is very restrictive and often impossible to realize in practice. It is therefore natural to ask what kind of identifiability results one can get when this assumption is dropped. Concretely, this question can be addressed by considering the set of operators \begin{equation*}
\opclassbr \defeq \bigcup_{M_\Gamma\colon \area (M_\Gamma) \leq \SSF} \mc H_{M_\Gamma}
\end{equation*}
 which consists of all sets $\opclassr$ such that $\area (M_\Gamma) \leq \SSF$. 
 
Our main identifiability results are stated in the two theorems below. 
\begin{theorem}
The set of operators  $\opclassbr$ is identifiable if and only if 
$\SSF \leq \sfrac{1}{2}$.
\label{pr:identifiabilitynecsparse}
\label{th:sufficiency}
\end{theorem}
\begin{proof}
See Section \ref{sec:proofmthm}. 
\end{proof}

The main implication of Theorem \ref{th:sufficiency} is that the penalty for not knowing the spreading function's support region prior to identification is a factor-of-two in the area of the spreading function. 
The origin of this factor-of-two penalty can be elucidated as follows. 
For operators $H_1,H_2$ with spreading function supported on $M_\Gamma$, i.e., $H_1, H_2 \in \opclassr$, we have $(H_1-H_2) \in \opclassr$, i.e.,\footnote{Homogeneity is trivially satisfied.} $\opclassr$ is a linear subspace of $\opclass$. 
In the case of unknown spreading function support region we have to deal with the (much larger) set $\opclassbr$, consisting of all sets $\opclassr$ with $\area (M_\Gamma) \leq \SSF$. It is readily seen that $\opclassbr$ is not a linear subspace of $\opclass$. Simply take $H_1,H_2 \in \opclassbr$ such that the support regions of $\sfuncm{1}$ and $\sfuncm{2}$ both have area $\SSF$ and are disjoint. While $(H_1-H_2) \notin \opclassbr$, we do, however, have that $H_1-H_2 \in \mathcal X(2 \SSF)$. 
This observation lies at the heart of the factor-of-two penalty in $\SSF$ as quantified by Theorem \ref{th:sufficiency}. 

We can eliminate this penalty by relaxing the identification requirement to apply to ``almost all'' $H\in \opclassbr$ instead of ``all'' $H\in \opclassbr$. 
\newcommand{\opsetaa}{\mc Y(\Delta)}
To be specific, we consider identifiability of a subset $\opsetaa \subset \opclassbr$, containing ``almost all'' $H \in \opclassbr$. The set $\opsetaa$ is obtained as follows. First, set 
\begin{align}
s_{k,m}(t,f) \defeq \sfunc \! \left( t + k  \csamp , f +   \frac{m}{\csamp L}\right)  e^{j 2 \pi \left( f + \frac{m}{\csamp L} \right)t}
\label{def:skm}
\end{align}
for $(k,m) \in \Gamma, \; (t,f) \in \Un$, 
and then define 
\begin{align*}
\opsetaa \defeq \{H \in \opclassbr \colon \{s_{k,m}(t,f), (k,m) \in \Gamma\} 
 \text{ are linearly independent on } \Un\}.
\end{align*}
The motivation for this specific definition of the set $\opsetaa$ will become clear in Section \ref{sec:almostall}. 
At this point, it is only important to note that the condition on the $s_{k,m}(t,f)$ in the definition of $\opsetaa$ allows to eliminate the factor-of-two penalty in $\SSF$. 
\begin{theorem} 
The set of operators $\opsetaa$ is identifiable if $\SSF < 1$. 
\label{th:almost_all}
\end{theorem}
\begin{proof}
See Section \ref{sec:almostall}.
\end{proof}
In order to demonstrate that ``almost all'' $H\in \opclassbr$ are in $\opsetaa$, suppose that
\footnote{Note that every $H\in \opclassbr$ can be represented by an expansion of the corresponding $s_{k,m}(t,f)$ into a set of orthonormal functions.} $s_{k,m} (t,f) = \sum_{p=1}^{P} c_{p}^{(k,m)} g_p(t,f)$, where $\{g_p(t,f), p=1,...,P\}$ is a set of functions orthogonal on $\Un$ and the $c_{p}^{(k,m)}$ are drawn independently from a continuous distribution. Then, the $s_{k,m}(t,f)$ will be linearly independent on $\Un$ with probability one \cite{eaton_non-singularity_1973}, if $P\geq L$. 
Finally, note that the operator $H$ with spreading function
\[
\sfunc \! \left( t + k  \csamp , f +   \frac{m}{\csamp L}\right) = e^{-j 2 \pi  \frac{m}{\csamp L} t}, \quad (k,m) \in \Gamma
\]
 where $\Gamma$ satisfies $|\Gamma|/L \leq \SSF$, is an example of an operator that is in $\opclassbr$ but not in $\opsetaa$. 

Putting things together, we have shown that ``almost all'' operators $H\in \opclassbr$ can be identified if $\SSF < 1$. 
This result is surprising as it says that there is no penalty for not knowing the spreading function's support region prior to identification, provided that one is content with a recovery guarantee for ``almost all'' operators. 


The factor-of-two penalty in Theorem \ref{th:sufficiency} has the same roots as the factor-of-two penalty in sparse signal recovery \cite{donoho_optimally_2003}, in the recovery of sparsely corrupted signals \cite{studer_recovery_2011}, in the recovery of signals that lie in a union of subspaces \cite{lu_theory_2008}, and, most pertinently, in
spectrum-blind sampling as put forward by Feng and Bresler \cite{feng_spectrum-blind_1996,feng_universal_1997,bresler_spectrum-blind_2008} and Mishali and Eldar \cite{mishali_blind_2009}. 
We hasten to add that Theorem \ref{th:almost_all} is inspired by the insight that---in the context of spectrum-blind sampling---the factor-of-two penalty in sampling rate can be eliminated by relaxing the recovery requirement to ``almost all'' signals \cite{feng_spectrum-blind_1996,bresler_spectrum-blind_2008}. Despite the conceptual similarity of the statement in Theorem \ref{th:almost_all} above and the result in \cite{feng_spectrum-blind_1996,bresler_spectrum-blind_2008}, the technical specifics are vastly different, as we shall see later. 

\paragraph*{Generalizations}

Theorems \ref{th:sufficiency} and \ref{th:almost_all} can easily be extended to operators with multiple inputs (and single output), i.e., operators whose response to the vector-valued signal 
 $\vect{\prsig}(t) =  \transp{[ \prsig_{0}(t), ..., \prsig_{\I-1}(t) ]}$ is given by 
\begin{equation}
(\vect{H} \vect{x}) (t) =    \sum_{\i=0}^{\I-1}  \int_\tau \int_\nu  \sfuncm{\i} (\tau,\nu) \prsig_{\i}(t-\tau)   e^{j 2\pi \nu t}  d\nu d\tau
\label{eq:ltvsysmiso}
\end{equation}
where $\sfuncm{\i}(\tau,\nu)$ is the spreading function corresponding to the (single-input) operator between input $\i$ and the output. 
For the case where the support regions of all spreading functions $\sfuncm{i}$ are known prior to identification, Pfander showed in \cite{pfander_measurement_2008} that the operator $\vect{H}$ is identifiable if and only if $\sum_{\i=0}^{\I-1} \area(\supp(\sfuncm{i})) < 1$. When the support regions are unknown, 
an extension of Theorem \ref{th:sufficiency} shows that $\vect{H}$ is identifiable if and only if  $\sum_{\i=0}^{\I-1} \area(\supp(\sfuncm{i})) \leq 1/2$. 
If one asks for identifiability of ``almost all'' operators only, the condition $\sum_{\i=0}^{\I-1} \area(\supp(\sfuncm{i})) \leq 1/2$ is replaced by $\sum_{\i=0}^{\I-1} \area(\supp(\sfuncm{i})) < 1$. Finally, we note that these results carry over to the case of operators with multiple inputs and multiple outputs (MIMO). Specifically, a MIMO channel is identifiable if each of its MISO subchannels is identifiable, 
see \cite{pfander_measurement_2008} for the case of known support regions.

\section{\label{sec:proofmthm}Proof of Theorem \ref{pr:identifiabilitynecsparse}}

\subsection{Necessity \label{sec:necessary}}

To prove necessity in Theorem \ref{pr:identifiabilitynecsparse}, we start by stating an equivalence condition on the identifiability of $\opclassbr$. This condition is often easier to verify than the condition in Definition \ref{def:defstableidentification}, and is inspired by a related result on sampling of signals in unions of subspaces  \cite[Prop. 2]{lu_theory_2008}. 
\begin{lemma}
$\prsig$ identifies $\opclassbr$ if and only if it identifies all sets
\[
\mc H_{M_\Phi \cup M_\Theta} \defeq \{H\colon H = H_1 - H_2 , H_1 \in \mc H_{M_\Phi}, H_2 \in \mc H_{M_\Theta} \}
\]
 with $\area(M_\Phi) \leq \SSF$ and $\area(M_\Theta) \leq \SSF$, where $\Phi, \Theta \subseteq \Sigma$. 
\label{pr:blind_identification}
\end{lemma}
\begin{proof}
First, note that the set of differences of operators in $\opclassbr$ can equivalently be expressed as
\begin{align}
\{H \colon H = H_1-H_2, \;H_1, H_2\in \opclassbr\}  
= \bigcup_{M_\Phi, M_\Theta \colon \area(M_\Phi), \area(M_\Theta) \leq \SSF}  \mc H_{M_\Phi \cup M_\Theta}.
\label{eq:eqsetdiff} 
\end{align}
From Definition \ref{def:defstableidentification} it now follows that $\prsig$ identifies $\opclassbr$ if there exist constants $0<\alpha \leq \beta < \infty$ such that for all $H \in \bigcup_{M_\Phi, M_\Theta \colon \area(M_\Phi), \area(M_\Theta) \leq \SSF}  \mc H_{M_\Phi \cup M_\Theta}$ 
we have
\begin{equation}
\alpha \norm[\opclass]{H} \leq \norm[]{ H \prsig  } \leq \beta \norm[\opclass]{ H }.
\label{eq:normineqinlem}
\end{equation}
Next, note that for $H_1,H_2 \in \mc H_{M_\Phi \cup M_\Theta}$, we have that $H_1-H_2 \in \mc H_{M_\Phi \cup M_\Theta}$. We can therefore conclude that \eqref{eq:normineqinlem} is equivalent to 
\begin{align}
\alpha \norm[\opclass]{H_1-H_2 } \leq \norm[]{ H_1 \prsig - H_2 \prsig  } \leq \beta \norm[\opclass]{ H_1-H_2 }
\label{eq:somewinlem1}
\end{align}
for all $H_1, H_2 \in \mc H_{M_\Phi \cup M_\Theta}$, and for all $M_\Phi$ and $M_\Theta$ with $ \area(M_\Phi), \area(M_\Theta) \leq \SSF$. Recognizing that \eqref{eq:somewinlem1} is nothing but saying that $\prsig$ identifies $\mc H_{M_\Phi \cup M_\Theta}$ for all $M_\Phi$ and $M_\Theta$ with $ \area(M_\Phi), \area(M_\Theta) \leq \SSF$, the proof is concluded. 
\end{proof}

Necessity in Theorem  \ref{th:sufficiency} now follows by choosing $M_\Phi,M_\Theta$ such that $M_\Phi \cap M_\Theta = \emptyset$ and $\area(M_\Phi)=\area(M_\Theta)=\SSF>1/2$. 
This implies $\area ( M_\Phi \cup M_\Theta  ) > 1$ and hence application of Theorem \ref{pr:necessity} to the corresponding set $\mc H_{M_\Phi \cup M_\Theta}$ establishes that $\mc H_{M_\Phi \cup M_\Theta}$ is not identifiable. By Lemma \ref{pr:blind_identification} this then implies that $\opclassbr$ is not identifiable.


\subsection{Sufficiency \label{sec:suff}}

We provide a constructive proof of sufficiency by finding a probing signal $\prsig$ that identifies $\opclassbr$, and showing how $\sfunc$ can be obtained from $H\prsig$. 
Concretely, we take $\prsig$ to be a weighted $\csamp L$-periodic train of Dirac impulses
\begin{equation}
\prsig(t) = \sum_{k \in \mb Z}  c_k  \delta(t  +   k\csamp ), \quad c_k = c_{k+L}, \; \forall k \in \mb Z.
\label{eq:probingsignal}
\end{equation}
The specific choice of the coefficients $\mbf c = \transp{[c_0,...,c_{L-1}]}$ 
will be discussed later. 

Kailath \cite{kailath_measurements_1962} and Kozek and Pfander \cite{kozek_identification_2005} used an unweighted train of Dirac impulses as probing signal to prove that LTV systems with spreading function compactly supported on a  rectangle (known prior to identification) of area $\SSF\leq 1$ are identifiable. Pfander and Walnut \cite{pfander_measurement_2006} used the probing signal \eqref{eq:probingsignal} to prove the result reviewed as Theorem \ref{pr:necessity} in this paper. Using a {\it weighted} train of Dirac impulses will turn out crucial  in the case of unknown spreading function support region, as considered here. 
It was shown recently \cite[Thm.~2.5]{krahmer_local_2012} that identification in the case of known support region, i.e., for $\opclassr$, is possible only with probing signals that decay neither in time nor in frequency, making Dirac trains a natural choice. 

The main idea of our proof is to i)   
reduce the identification problem to that of solving a continuously indexed linear system (of $L$ equations in $L^2$ unknowns), and ii) based on Lemma \ref{pr:blind_identification} to show that the solution of this underdetermined linear system of equations is unique whenever $\SSF \leq \sfrac{1}{2}$, provided that $\mbf c$ is chosen appropriately. 

We start by computing the response of $H$ to $\prsig(t)$ in \eqref{eq:probingsignal}. From \eqref{eq:hilbertschmidt_s} we get
\begin{align}
y(t) 	&= (H \prsig) (t)  
	        =  \sum_{k \in \mb Z}  c_k \int_\nu   \sfunc(t + k \csamp  ,\nu)  e^{j 2 \pi \nu t} d\nu. \label{eq:dirac}	        
\end{align}
Next, we use the Zak transform \cite{janssen_zak_1988} to turn \eqref{eq:dirac} into a continuously indexed linear system of equations as described in Step i) above. The Zak transform (with parameter $\csamp L$) of the signal $y(t)$ is defined as
\begin{equation}
\mc Z_{y}(t ,f)   \nonumber \defeq \sum_{m \in \mb Z} y(t - m \csamp L ) e^{j2\pi  m \csamp L f} 
\end{equation}
for $(t,f) \in  [0, \csamp L)  \times [0,1/(\csamp L))$,  
and satisfies the following (quasi-)periodicity properties 
\begin{align*}
\mc Z_{y}(t +\csamp L ,f)          &= e^{j 2 \pi  \csamp L f} \mc Z_{y}(t ,f),\\
\mc Z_{y}(t ,f  + 1/(\csamp L) ) &= \mc Z_{y}(t ,f).
\end{align*}
It is therefore sufficient to consider $\mc Z_{y}(t ,f)$ on the fundamental rectangle $[0, \csamp L)  \times [0,1/(\csamp L))$. The Zak transform is an isometry, i.e., it satisfies
\begin{equation}
\csamp L \int_0^{\csamp L}\!\!\!\! \int_0^{1/(\csamp L)} |\mc Z_{y}(t ,f)|^2 = \norm{y}^2.
\label{eq:unitarityzak}
\end{equation}
The Zak transform of $y(t)$ in \eqref{eq:dirac} is given by
\begin{align}
\mc Z_{y} (t,f) 
&=\!\!\! \sum_{k,m \in \mb Z} \!\!\! c_k  \!\! \int_\nu \!\!  \sfunc(t - m \csamp L+k \csamp ,\nu)  e^{j 2 \pi \nu (t -m \csamp L)} d\nu \hspace{0.045cm} e^{j2\pi m \csamp L  f}  \nonumber \\
&= \! \! \sum_{k' \in \mb Z}  \!c_{k'} \!\!  \int_\nu \!\!\!  \sfunc(t+ k' \csamp ,\nu)  e^{j 2 \pi \nu t}  \!  \sum_{m \in \mb Z} \! e^{-j 2 \pi (\nu-f) m \csamp L}  d\nu \hspace{0.01cm} \, \label{eq:substkd2} \\
&= \! \! \sum_{k \in \mb Z} \! c_{k} \!\!  \int_\nu \!\!\!  \sfunc(t+ k \csamp ,\nu)  e^{j 2 \pi \nu t}    \frac{1}{TL}\! \sum_{m\in \mb Z}\! \delta\left(\nu \!- \! \left(f \! + \! \frac{m}{TL}\right) \right)  d\nu \hspace{0.01cm} \, \label{eq:substkd3} \\
&= \sum_{k \in \mb Z}   \frac{c_{k}}{\csamp L} \sum_{m \in \mb Z} \sfunc \!  \left(t+ k \csamp, f + \frac{m}{\csamp L} \right)  e^{j 2 \pi t \left( f + \frac{m}{\csamp L} \right) } \label{eq:lstepcktl}
\end{align}
where we used the substitution $k'= k-mL$ in \eqref{eq:substkd2} and \eqref{eq:substkd3} follows from $\sum_{m \in \mb Z} e^{-j 2 \pi (\nu-f) m \csamp L} = \frac{1}{TL} \sum_{m\in \mb Z} \delta\left(\nu - \left(f+\frac{m}{TL}\right) \right)$. 
Next, we split the fundamental rectangle $[0, \csamp L)  \times [0,1/(\csamp L))$ into $L$ cells $\Un$, where $\Un = \left[ 0, \csamp  \right) \times \left[0, 1/(\csamp L)  \right)$  was defined in Section \ref{sec:prform} in the context of  structural assumptions imposed on the spreading function. Concretely, we substitute $t = t'+ p\csamp$ in \eqref{eq:lstepcktl}, with $p \in \{0,...,L-1\}$ and $t' \in [0,\csamp)$. This yields, for $(t',f) \in \Un$ and $p = 0,...,L-1$,
\begin{align}
z_p(t',f)  
&\defeq  \mc Z_{y}( t'+ p\csamp ,f)  \label{eq:defpz}\\
&=  \sum_{k \in \mb Z}   \frac{c_{k} }{\csamp L}   \sum_{m\in \mb Z} \!\!  \sfunc \! \left(t'+p \csamp  + k \csamp  , f + \frac{m}{\csamp L}  \right) e^{j 2 \pi (t'+ p \csamp ) \left(f + \frac{m}{\csamp L} \right) } \nonumber \\
&=  \sum_{k' \in \mb Z}  \frac{c_{k'-p} }{\csamp L}  \sum_{m \in \mb Z}  \sfunc \! \left(t'+ k' \csamp , f + \frac{m}{\csamp L}  \right) e^{j 2 \pi (t'+p \csamp) \left(f + \frac{m}{\csamp L} \right) } \nonumber \\
&= \sum_{k =0}^{L-1} \frac{c_{k-p} }{\csamp L}   \sum_{m = 0}^{L-1} \sfunc \! \left( t'+ k \csamp, f + \frac{m}{\csamp L}  \right) e^{j 2 \pi (t'+ p \csamp) \left(f + \frac{m}{\csamp L}  \right) } \label{eq:lszptdf}
\end{align}
where 
\eqref{eq:lszptdf} is a consequence of $\sfunc(\tau , \nu ) = 0$ for $(\tau,\nu) \notin \left[0, \csamp L \right) \times \left[0, 1/\csamp \right)$, by assumption. 
We next rewrite \eqref{eq:lszptdf} in vector-matrix form. To this end, we define the column vectors $\mbf z(t,f)$ and $\mbf s(t,f)$ according to
\begin{equation}
[\mbf z(t,f)]_p \defeq  \csamp L\, z_{p}(t,f)e^{-j2\pi p \csamp f} , \quad p=0,...,L-1 
\label{eq:defzvect}
\end{equation}
and
$\mbf s(t,f)  \defeq [s_{0,0}(t,f),s_{0,1}(t,f),\allowbreak  ...,\allowbreak  s_{0,L-1}(t,f), \allowbreak s_{1,0}(t,f),..., s_{L-1,L-1}(t,f)]^T$
with $s_{k,m}(t,f)$ as defined in \eqref{def:skm}. 
Since $\sfunc(\tau , \nu ) = 0$ for $(\tau,\nu) \notin \left[0, \csamp L \right) \times \left[0, 1/\csamp \right)$, the vector $\mbf s(t,f)$, $(t,f) \in \Un$, fully characterizes the spreading function $\sfunc(\tau,\nu)$. 
With these definitions \eqref{eq:lszptdf} can now be written as
\begin{equation}
\mbf z(t,f) =   \mbf A_{\mbf c}  \mbf s(t,f), \quad  (t,f) \in  U
\label{eq:sysofeq}
\end{equation}
with the $L\times L^2$ matrix
\begin{equation}
\mbf A_{\mbf c} \defeq  [\mbf A_{\mbf c, 0} | \; ...  \; | \mbf A_{\mbf c, L-1}  ], \quad
\mbf A_{\mbf c, k} \defeq  \mbf C_{\mbf c, k} \herm{\mbf F}
\label{eq:defAc}
\end{equation}
where 
$[\mbf F]_{p,m} =e^{-j2\pi \frac{pm}{L}}, \; p,m=0,...,L-1$, and $\mbf C_{\mbf c, k}$ is the $L\times L$ diagonal matrix with diagonal entries $\{c_{k}, c_{k-1}, ..., c_{k -L+1}\}$ (recall that the coefficient sequence $c_k$ is $L$-periodic).  

Since $\mbf z(t,f)$ is obtained from the operator's response to the probing signal and $\mbf s(t,f)$ fully determines the spreading function $\sfunc$, we can conclude that the identification of $H$ has been reduced to the solution of a continuously indexed linear system of equations. Conceptually, for each pair $(t,f) \in \Un$, we need to solve a linear system of $L$ equations in $L^2$ unknowns. 
The proof is then completed by showing that this continuously indexed linear system of equations has a unique solution if $\SSF \leq 1/2$. 
More formally, we need to relate identifiability according to Definition \ref{def:defstableidentification} to solvability of the continuously indexed linear system of equations \eqref{eq:sysofeq}. 
To this end, we first note that thanks to Lemma \ref{pr:blind_identification}, it suffices to prove identifiability of $\mc H_{M_\Phi \cup M_\Theta}$ for all pairs $M_\Phi,M_\Theta$ with 
$\area(M_\Phi) \leq \sfrac{1}{2}$ and $\area(M_\Theta) \leq \sfrac{1}{2}$. By setting $M_\Gamma = M_\Phi \cup M_\Theta$ this is equivalent to proving identifiability of $\mc H_{M_\Gamma}$ for all $M_\Gamma$ with $\area(M_\Gamma)\leq 1$. 
For $H\in \opclassr$, by definition, $s_{k,m}(t,f) = 0, \; \forall  (k,m) \notin \Gamma$.  Denote the restriction of the vector $\mbf s(t,f)$ to the entries corresponding to the active cells, i.e., the cells indexed by $\Gamma$, by $\mbf s_\Gamma(t,f)$ and let $\mbf A_\Gamma$ be the matrix containing the columns of  $\mbf A_{\mbf c}$ that correspond to the index set $\Gamma$. The linear system of equations \eqref{eq:sysofeq} then reduces to
\begin{equation}
\mbf z(t,f) =  \mbf A_{\Gamma}  \mbf s_\Gamma(t,f), \quad (t,f) \in  \Un.
\label{eq:sysofeq_rest}
\end{equation} 
Solvability of \eqref{eq:sysofeq_rest} can now formally be related  to identifiability through the following lemma, proven in Appendix \ref{app:lem2}. 
\begin{lemma}
Let $\prsig$ be given by \eqref{eq:probingsignal}. Then, the (tightest) bounds $\alpha, \beta$ in \eqref{eq:cond_stable} for the set of operators $\opclassr$ are given by
\begin{equation}
 \alpha_\Gamma = \frac{1}{\sqrt{TL}} \inf_{\norm[2]{\mbf v} = 1} \norm[2]{ \mbf A_{\Gamma} \mbf v} \;\text{ and } \;\beta_\Gamma = \frac{1}{\sqrt{TL}} \sup_{\norm[2]{\mbf v} = 1} \norm[2]{ \mbf A_{\Gamma} \mbf v}.
\label{eq:stabbounds}
\end{equation}
\label{le:boundedness_classr}
\end{lemma}

The proof of sufficiency in Theorem \ref{th:sufficiency} is now completed by showing that for all $M_\Gamma$ with $\area(M_\Gamma)\leq 1$, $\opclassr$ is identifiable, i.e., 
$0<\alpha_\Gamma \leq \beta_\Gamma < \infty$. 
By Lemma \ref{le:boundedness_classr}, $\beta_\Gamma < \infty$ trivially, and showing that $\alpha_\Gamma>0$ amounts to proving that $\mbf A_\Gamma$ has full rank for all $\Gamma \subseteq \Sigma$ with $|\Gamma|\leq L$, i.e., for all $M_\Gamma$ such that $\area(M_\Gamma)\leq 1$.  
What comes to our rescue is \cite[Thm. 4]{lawrence_linear_2005} which states that for almost all $\mbf c$, each\footnote{
Pfander and Walnut \cite{pfander_measurement_2006} used the probing signal \eqref{eq:probingsignal} to prove that, for known spreading function support region, $\SSF \leq 1$ is sufficient for identifiability. The  crucial difference between \cite{pfander_measurement_2006} and our setup is that we need {\it each} submatrix of $L$ columns  of $\mbf A_{\mbf c}$ to have full rank, as we do not assume prior knowledge of the support region. 
}
 $L\times L$ submatrix of $\mbf A_{\mbf c}$ has full rank. 
In the remainder of the paper $\mbf c$ is chosen such that each $L\times L$ submatrix of $\mtx{A}_{\ve{c}}$, indeed, has full rank. In other words, $\mbf c$ is chosen such that $\spark(\mtx{A}_{\ve{c}}) = L+1$.

\subsection{\label{sec:relspecbl}Relation to spectrum-blind sampling}

The philosophy of operator identification without prior knowledge of the spreading function's support region is related to the idea of spectrum-blind sampling of multi-band signals \cite{feng_spectrum-blind_1996,feng_universal_1997,lu_theory_2008,mishali_blind_2009}. In spectrum-blind sampling the central problem is to recover a signal, sparsely supported on a priori unknown frequency bands, from its samples taken at a rate that is (much) smaller than the Shannon-Nyquist rate of the signal. 
 The conceptual relation between operator identification and spectrum-blind sampling is brought out by comparing \eqref{eq:sysofeq} to the recovery equation in spectrum-blind sampling, given by \cite{feng_spectrum-blind_1996,feng_universal_1997,lu_theory_2008,mishali_blind_2009}
 \begin{align}
 \vect{y}(f) = \vect{A} \vect{x}(f),\quad f\in \mathcal F.
 \label{eq:specblindreconst}
 \end{align} 
 Here $\vect{A}\in \complexset^{m\times n}$, with $m < n$, depends on the sampling pattern, $\vect{x}(f), f\in \mathcal F$, fully specifies the signal to be reconstructed, and $\vect{y}(f), f\in \mathcal F$, is obtained from the samples of the signal. Further, $\mathcal F$ is a spectral ``cell'', playing a role similar to the cell $\Un$ in our setup.  
It is shown in \cite{feng_spectrum-blind_1996,feng_universal_1997,lu_theory_2008,mishali_blind_2009} that the penalty for not knowing the spectral support set is a factor-of-two in sampling rate. The corresponding result in the present paper is Theorem \ref{th:sufficiency}. 
It is furthermore shown in \cite{feng_spectrum-blind_1996,bresler_spectrum-blind_2008} that there is no penalty for not knowing the spectral support set if one requires recovery of \emph{almost all} signals only. The corresponding result in this paper is Theorem \ref{th:almost_all}. 
Despite this strong structural similarity, there is a fundamental difference between spectrum blind sampling and the system identification problem considered here. In operator identification a function of \emph{two} variables, $\sfunc(\tau,\nu)$, has to be extracted from the \emph{univariate} measurement $(H\prsig)(t)$. Moreover, in spectrum-blind sampling there is no limit on the cardinality of the spectral support set that would parallel the $\SSF \leq 1/2$ or $\SSF < 1$ thresholds. 

\section{Recovering the spreading function \label{sec:reconstruct} }
We next present an algorithm that provably recovers all $H\in \opclassbr$ for $\SSF\leq 1/2$ from the operator's response $H\prsig$ to the probing signal $\prsig(t)$ in \eqref{eq:probingsignal}. 
The algorithm first identifies the support set of $\sfunc(\tau,\nu)$, i.e., the index set $\Gamma$, and then solves the corresponding linear system of equations \eqref{eq:sysofeq_rest}, which, based on \eqref{def:skm}, yields $\sfunc(\tau,\nu)$. 
Starting from \eqref{eq:sysofeq_rest}, an explicit reconstruction formula for $\sfunc(\tau,\nu)$ is straightforward to derive and is given by
\begin{align}
\sfunc \! \left( t + k  \csamp , f +   \frac{m}{\csamp L}\right)  = \csamp L \sum_{p=0}^{L-1} [\pinv{ \mbf A}_\Gamma]_{l,p}  
\, z_{p}(t,f)e^{-j2\pi  \left( p \csamp f   + \left(f+ \frac{m}{\csamp L} \right)t \right) }
\label{eg:endpiecwise}
\end{align}
for $(k,m) \in \Gamma$, $(t,f)\in \Un$, where $\pinv{\mbf A}_\Gamma$ is the pseudoinverse of $ \mbf A_\Gamma$ and the index $l$ refers to the row of $\pinv{\mbf A}_\Gamma$ corresponding to the $(k,m)$th cell. 

\newcommand{\IMV}[0]{(\text{P0})} 
\newcommand{\IMVM}[0]{(\text{$\overline{\text{P0}}$})} 
\newcommand{\MMV}[0]{(\text{P0$'$})}	
\newcommand{\MMVN}[0]{(\widetilde{\text{P0}})}
\newcommand{\Bs}[1]{\mtx{B}_{{#1}}} 

We now turn our attention to the main challenge, namely support set recovery. 
Formally, \eqref{eq:sysofeq} is a continuously indexed linear system of equations, whose solutions (across indices $(t,f) \in \Un$) share the support set $\Gamma$. This problem was studied before under the name of ``infinite measurement vector problem'' in \cite{mishali_reduce_2008} as a generalization of the multiple measurement vector (MMV) problem \cite{jie_chen_theoretical_2006}, where the reconstruction of a finite number of vectors sharing a sparsity pattern, from a finite number of linear measurements, is considered. 
Starting from the observation that the cardinality of the index set $\Gamma$ is finite, and the matrix $\vect{A}_\vect{c}$ is finite-dimensional, it is perhaps not surprising to see that the infinite measurement vector problem at hand can be reduced to an MMV problem. 
Based on the recovery equation \eqref{eq:specblindreconst}, this was recognized before in the context of spectrum-blind sampling in \cite{feng_spectrum-blind_1996,mishali_blind_2009,bresler_spectrum-blind_2008} and, in a more general context, in \cite{mishali_reduce_2008}. 
We next present a general reduction method, which unifies the approaches in \cite{feng_spectrum-blind_1996,mishali_blind_2009,mishali_reduce_2008,bresler_spectrum-blind_2008} and is based on a simplified, and, as we believe, more accessible treatment. The discussion  in Section \ref{sec:redtommv} below is therefore of interest in its own right. 

We assume throughout that $|\Gamma| \leq L$; this is w.l.o.g. as $|\Gamma| \leq L$ corresponds to $\SSF \leq 1$ and we only consider the identification of operators satisfying $\SSF \leq 1/2$ or $\SSF < 1$. 
The index set $\Gamma$ can be recovered as follows: 
\[
\IMV \; \begin{cases} 
		\text{minimize} 	& |\Gamma| \\
		\text{subject to} 	& \mbf z(t,f) = \mbf A_{\Gamma}  \mbf s_\Gamma(t,f), \quad (t,f) \in \Un,
\end{cases}
\]
where the minimization is performed over all $\Gamma \subseteq \Sigma$ and all corresponding $\mbf s_\Gamma(t,f)\colon \Un^{|\Gamma|} \to \complexset$.

\subsection{\label{sec:redtommv}Reduction to an MMV problem}
The proof of $\IMV$ delivering the correct solution is deferred to Section \ref{se:uniqueness}. We first develop a unified approach to the reduction of the infinite measurement vector problem $\IMV$ to an MMV problem. 
We emphasize, as mentioned before, that this reduction approach encompasses the settings in \cite{feng_spectrum-blind_1996,mishali_blind_2009,mishali_reduce_2008,bresler_spectrum-blind_2008} and hence applies to spectrum-blind sampling, inter alia. 
Our approach is based on a basis expansion of the elements of $\mbf z(t,f)$ and $\mbf s_\Gamma(t,f)$. We start with some definitions.  
Consider the linear space of functions $\mathcal G = \{g(t,f)\colon \Un \to \complexset \}$ equipped with the inner product $\innerprod{g_1}{g_2} = \int_\Un g_1(t,f) \conj{g}_2(t,f) d(t,f)$, $g_1,g_2 \in \mathcal G$, and induced norm $\norm{g} = \sqrt{\innerprod{g}{g}}$. 
Let $\{b_0(t,f),\allowbreak...,\allowbreak b_{K-1}(t,f) \in \mathcal G \}$ be a basis (not necessarily orthogonal) for the space spanned by the functions 
$\{[\mbf z(t,f)]_p, p=0,...,L-1\}$ and set $K =\dim \spanof \{[\mbf z(t,f)]_p, p  =0,...,L-1\}$. We can represent $\vect{z}(t,f)$  in terms of the basis elements $b_i(t,f)$ according to
\begin{equation}
\mbf z(t,f) =  \mtx{B}_{z} \vect{b}(t,f)
\label{eq:relzbb}
\end{equation}
where $\vect{b}(t,f) \defeq \transp{[b_0(t,f),...,b_{K-1}(t,f)]}$ and $\mtx{B}_{z} \in \mathbb C^{L\times K}$ contains the expansion coefficients of $\mbf z(t,f)$ in the basis $\{b_0(t,f),...,b_{K-1}(t,f)\}$. 
It follows from $K =\dim \spanof \{ [\mbf z(t,f)]_p, p  =0,...,L-1 \}$ that $\mtx{B}_{z}$ has full rank $K\leq L$. 
To see this, suppose that $\rank{\mtx{B}_{z}} < K$. Then, each set of $K$ rows of $\mtx{B}_{z}$ is linearly dependent, i.e., for each set of rows of  $\mtx{B}_{z}$, indexed by say $\Phi$,  with cardinality $|\Phi|=K$, there exists an $\vect{a}\in \mathbb C^K$, $\vect{a}\neq \vect{0}$, such that 
\begin{equation}
\herm{\vect{a}} \mtx{B}_{z}^\Phi = \vect{0}
\label{eq:ambz}
\end{equation}
where $\mtx{B}_{z}^\Phi$ is the matrix obtained by retaining the rows of $\mtx{B}_{z}$ in $\Phi$. Then, for each $\Phi$ with $|\Phi| = K$, according to \eqref{eq:ambz}, there exists an $\vect{a}\neq \vect{0}$ such that  
\[
\herm{\vect{a}} \mtx{B}_{z}^\Phi \vect{b}(t,f) = \herm{\vect{a}} \mbf z_\Phi(t,f)   = 0 
\]
where $\mbf z_\Phi(t,f)$ contains the entries of $\mbf z(t,f)$ corresponding to the index set $\Phi$. 
This would, however, imply $\dim \spanof \{[\mbf z(t,f)]_p, p  =0,...,L-1 \} < K$, which stands in contradiction to $\dim \spanof \{\allowbreak [\mbf z(t,f)]_p, p  =0,...,L-1 \} = K$. 
%

Expanding $\vect{s}_\Gamma(t,f)$ in \eqref{eq:sysofeq_rest} in the basis\footnote{
Thanks to $\mbf A_\Gamma$ having full column rank, $\spanof\{s_{k,m}(t,f) \colon (k,m) \in \Gamma \} \subseteq \spanof \{b_k(t,f), k=0,...,K-1\}$ (cf.~\eqref{eq:sysofeq_rest} and \eqref{eq:relzbb}). 
} $\{b_0(t,f),...,\allowbreak b_{K-1}(t,f)\}$, we can rewrite the constraint in $\IMV$ as
\begin{equation}
 \mtx{B}_{z} \vect{b}(t,f) = \mbf A_{\Gamma} \Bs{\Gamma}
 \vect{b}(t,f), \quad (t,f) \in  \Un
\label{eq:sysofeq_inbasis}
\end{equation} 
where $\Bs{\Gamma} \in \mathbb C^{|\Gamma| \times K}$ contains the expansion coefficients of  $\mbf s_\Gamma(t,f)$ in the basis $\{b_0(t,f),...,\allowbreak b_{K-1}(t,f)\}$. 
Since the elements of $\vect{b}(t,f)$ form a basis, 
\eqref{eq:sysofeq_inbasis} 
is equivalent to 
\begin{equation}
\mtx{B}_{z} = \mbf A_{\Gamma} \Bs{\Gamma}.
\label{eq:sysofeq_incoeff}
\end{equation}
We have therefore shown that \IMV~is equivalent to 
\[
\MMVN \; \begin{cases} 
		\text{minimize} 	& |\Gamma| \\
		\text{subject to} 	& \mtx{B}_{z} =  \mbf A_{\Gamma} \Bs{\Gamma}
	\end{cases}
\]
where the minimization is performed over all $\Gamma \subseteq \Sigma$ and all corresponding $\Bs{\Gamma}  \in \mb C^{|\Gamma| \times K}$. 
We have thus reduced $\IMV$, which involves a continuum of constraints, to $\MMVN$, which involves only finitely many constraints. 
$\MMVN$ is known in the literature as the MMV problem \cite{jie_chen_theoretical_2006}, which is usually formulated equivalently as: minimize $\norm[\text{row-0}]{\mtx{B}_s}$ subject to $\mtx{B}_{z}=\mtx{A}_{\ve{c}} \mtx{B}_{s}$, where the constraint is over all $\mtx{B}_{s} \in \complexset^{L^2\times K}$ and $\norm[\text{row-0}]{\mtx{B}_{s}}$ is the number of non-zero rows of $\mtx{B}_{s}$.

We are now ready to explain the reduction approaches in \cite{feng_spectrum-blind_1996,feng_universal_1997,mishali_blind_2009,mishali_reduce_2008,bresler_spectrum-blind_2008} in the general reduction framework just introduced. We start with the method described in \cite{feng_spectrum-blind_1996,feng_universal_1997,mishali_blind_2009,bresler_spectrum-blind_2008} in the context of spectrum-blind sampling. This approach starts from a correlation matrix, which in our setup becomes 
\begin{align}
\mtx{C}_{z} \defeq \int_{\Un} \mbf z(t,f) \herm{\mbf z}(t,f) d(t,f).
\label{eq:corrmatxz}
\end{align} 
With \eqref{eq:sysofeq_rest} we can express $\mtx{C}_{z}$ as
\begin{equation}
\mtx{C}_{z}  =  \mbf A_{\Gamma} \mtx{C}_{\vect{s}_\Gamma}  \herm{\mbf A}_{\Gamma}
\label{eq:integral}
\end{equation}
where $\mtx{C}_{\vect{s}_\Gamma}  = \int_{\Un}   \mbf s_{\Gamma}(t,f)  \herm{\mbf s}_{\Gamma}(t,f) d( t,f)$. 
Analogously to the results in \cite[Sec. 3, Lem. 1]{feng_universal_1997} for signal recovery in spectrum-blind sampling, it can be shown that \IMV~is equivalent to
\[
\IMVM \; \begin{cases} 
		\text{minimize} 	& |\Gamma| \\
		\text{subject to} 	&  \mtx{C}_{z}  =  \mbf A_{\Gamma} \mtx{C}_{\vect{s}_\Gamma}   \herm{\mbf A}_{\Gamma}
\end{cases}
\]
 where the minimization is performed over all $\Gamma \subseteq \Sigma$ and all corresponding Hermitian $\mtx{C}_{\vect{s}_\Gamma}  \in \mb C^{|\Gamma| \times |\Gamma|}$. 

We next show that $\IMVM$ is equivalent to an MMV problem, and then explain this equivalence result in our basis expansion approach. 
 $\mtx{C}_{z}$ is a Hermitian matrix and can hence be decomposed as $\mtx{C}_{z}= \mbf Q \herm{\mbf Q}$ 
\cite[Thm.~4.1.5]{horn_matrix_1986}, where the $K = \rank{\mtx{C}_{z}}$ columns of $\mbf Q \in \mb C^{L \times K}$ are orthogonal. 
Analogously to \cite[Sec. 3, Lem. 1]{feng_universal_1997}, \cite[Sec. V-C]{mishali_blind_2009}, it can now be shown that \IMVM~(and by induction \IMV) is equivalent to the MMV problem 
\begin{equation*}
\MMV \; \begin{cases} 
		\text{minimize} 	& |\Gamma| \\
		\text{subject to} 	& \mbf Q = \mbf A_{\Gamma}  \mbf G_{\Gamma} 
\end{cases}
\end{equation*}
where the minimization is performed over all $\Gamma \subseteq \Sigma$ and all corresponding $\mbf G_{\Gamma} \in \mb C^{|\Gamma| \times K}$. 

To see how the reduction to $\MMV$ just described can be cast into the basis expansion approach described above, let $\mbf z(t,f) =  \mtx{B}_{z} \vect{b}(t,f)$, where $\vect{b}(t,f)$ is an orthonormal basis for $\spanof \{[\mbf z(t,f)]_p, p=0,...,L-1\}$. By \eqref{eq:corrmatxz}, we then have
\[
\mtx{C}_{z} =  \mtx{B}_{z} \left[ \int_{\Un} \vect{b}(t,f)  \herm{\vect{b}}(t,f) d(t,f) \right] \,\herm{\mtx{B}}_z  = \mtx{B}_{z} \herm{\mtx{B}}_z. 
\]
From 
$
\mtx{C}_{z} = \mtx{B}_{z} \herm{\mtx{B}}_z = \mtx{Q}\herm{\mtx{Q}}
$
it follows that there exists a unitary matrix $\mtx{U}$ such that $\mtx{B}_z = \mtx{Q}\vect{U}$, which is seen as follows. 
We first show that any solution $\mtx{B}$ to $\mtx{C}_{z} = \mtx{B} \herm{\mtx{B}}$ can be written as $\mtx{B} =\mtx{C}_{z}^{1/2}  \mtx{V} $, where $\mtx{V}$ is unitary \cite[Exercise on p. 406]{horn_matrix_1986}. 
Indeed, we have
\begin{align}
\mtx{I} 
&= \mtx{C}_{z}^{-1/2}  \mtx{C}_{z}^{1/2} \mtx{C}_{z}^{1/2} \mtx{C}_{z}^{-1/2} \nonumber \\
&= \mtx{C}_{z}^{-1/2}  \mtx{B} \herm{\mtx{B}} \mtx{C}_{z}^{-1/2} \nonumber \\
&= ( \mtx{C}_{z}^{-1/2} \mtx{B} ) \herm{( \mtx{C}_{z}^{-1/2} \mtx{B} )}
\label{eq:exhj}
\end{align}
where the last equality follows since $\mtx{C}_{z}^{-1/2}$ is self adjoint, according to \cite[Thm. 7.2.6]{horn_matrix_1986}. From  \eqref{eq:exhj} it is seen that $\mtx{V} \defeq\mtx{C}_{z}^{-1/2} \mtx{B}$ is unitary, and hence $\mtx{B} = \mtx{C}_{z}^{1/2} \mtx{V}$, with $\mtx{V}$ unitary. 
Therefore, we have $\mtx{B}_z = \mtx{C}_{z}^{1/2} \mtx{V}_1$ and $\mtx{Q} = \mtx{C}_{z}^{1/2} \mtx{V}_2$, where $\mtx{V}_1$ and $\mtx{V_2}$ are unitary, and hence $\mtx{B}_z = \mtx{Q} \herm{\mtx{V}}_2  \mtx{V}_1$. As $\herm{\mtx{V}}_2  \mtx{V}_1$ is unitary, we proved that there exists a unitary matrix $\mtx{U}$ such that $\mtx{B}_z = \mtx{Q}\vect{U}$. 
With $\mtx{B}_z = \mtx{Q}\vect{U}$, the minimization variable of $\MMVN$ is given by $\mbf B_\Gamma = \mbf G_\Gamma \mbf U$, where $\mbf G_\Gamma$ is the minimization variable of $\MMV$, hence $\MMVN$ and $\MMV$ are equivalent.

Another approach to reducing $\IMV$ to an MMV problem was put forward in \cite[Thm. 2]{mishali_reduce_2008}. 
In our setting and notation the resulting MMV problem is given by
\[
(\text{P0}'') \; \begin{cases} 
		\text{minimize} 	& |\Gamma| \\
		\text{subject to} 	&  \mbf W = \mbf A_{\Gamma} \mbf G_{\Gamma} 
\end{cases}
\]
where the minimization is performed over all $\Gamma \subseteq \Sigma$ and all corresponding $\mbf G_{\Gamma} \in \mb C^{|\Gamma| \times K}$. 
Here, the matrix $\mtx{W} \in \mb C^{L \times K}$ can be taken to be any matrix whose column span is equal to $\spanof \{\vect{z}(t,f)\colon (t,f) \in \Un \}$. 
To explain this approach in our basis expansion framework, we start by noting that \eqref{eq:relzbb} implies that $\spanof \{\vect{z}(t,f)\colon (t,f) \in \Un \} = \spanof ( \mtx{B}_z )$. We can therefore take $\mtx{W}$ to equal $\mtx{B}_z$. 
On the other hand,  for every $\mtx{W}$ with $\spanof(\mtx{W}) = \spanof \{\vect{z}(t,f)\colon (t,f) \in \Un \}$, we can find a basis $\vect{b}(t,f)$ such that $\mtx{W} \vect{b}(t,f)  = \vect{z}(t,f)$. 
Choosing different matrices $\mtx{W}$ in $(\text{P0}'')$ therefore simply amounts to choosing different bases $\vect{b}(t,f)$.

\subsection{Uniqueness conditions for $\IMV$ \label{se:uniqueness}}

We are now ready to study uniqueness conditions for $\IMV$. Specifically, we will find a necessary and sufficient condition for $\IMV$ to deliver the correct solution to the continuously indexed linear system of equations in \eqref{eq:sysofeq}. 
This condition comes in the form of a threshold on $|\Gamma|$ that depends on the ``richness'' of the spreading function, specifically, on $\dim \spanof \{ s_{k,m}(t,f), (k,m)\in \Gamma \}$. 
\begin{theorem}
Let $\mbf z(t,f) = \mbf A_{\Gamma}  \mbf s_\Gamma(t,f), \, (t,f) \in \Un$, with $\dim \spanof \{ s_{k,m}(t,f), (k,m)\in \Gamma \} = K$. Then $\IMV$ applied to $\mbf z(t,f)$ recovers  $(\Gamma,\mbf s_\Gamma(t,f))$ if and only if 
\begin{equation}
|\Gamma| < \frac{L+K}{2}.
\label{eq:unique_imv}
\end{equation}
\label{thm:uniqueness_IMV}
\end{theorem}
\vspace{-0.4cm}
%
Since $K\geq 1$, Theorem \ref{thm:uniqueness_IMV} guarantees exact recovery if $|\Gamma|\leq L/2$, and hence by $\area(M_\Gamma) = |\Gamma|/L$ (see Section \ref{sec:prform}), if $\SSF\leq 1/2$, 
which is the recovery threshold in Theorem \ref{th:sufficiency}. Recovery for $\SSF < 1$ will be discussed later. 
%
Sufficiency in Theorem \ref{thm:uniqueness_IMV} 
was shown in \cite[Prop. 1]{mishali_reduce_2008} and in the context of spectrum-blind sampling in \cite[Sec. 3, Thm. 3]{feng_universal_1997}. Necessity has not been proven formally before, but follows directly from known results, as shown in the proof of the theorem below.  

\begin{proof}[Proof of Theorem \ref{thm:uniqueness_IMV}]
The proof is based on the equivalence of $\IMV$ and $\MMVN$, established in the previous section, and on the following uniqueness condition for the MMV problem $\MMVN$. 
\begin{proposition}[
\cite{jie_chen_theoretical_2006,cotter_sparse_2005,wax_unique_1989,davies_rank_2010}] 
Let $\mtx{B}_{z} =  \mbf A_{\Gamma} \Bs{\Gamma}$ with $\rank{\Bs{\Gamma}} = K$. Then $\MMVN$ applied to $\mtx{B}_{z}$ recovers $(\Gamma, \Bs{\Gamma})$
if and only if
\begin{equation}
|\Gamma| < 
			\frac{L + K }{2}.
\label{eq:cond_uniq_mmv}
\end{equation}
\label{prop:cond_uniq_mmv}
\end{proposition}
\vspace{-0.4cm}
\begin{proof}[Proof of Proposition \ref{prop:cond_uniq_mmv}]
Sufficiency was proven in \cite[Thm. 1]{wax_unique_1989}, \cite[Lem. 1]{cotter_sparse_2005}, \cite[Thm. 2.4]{jie_chen_theoretical_2006}, necessity in \cite[Thm. 2]{davies_rank_2010}. We present a different, slightly simpler, argument for necessity in Appendix \ref{app:neccmmv}.
\end{proof}
In Section \ref{sec:redtommv}, we showed that $\dim \spanof \{ [\mbf z(t,f)]_p, p  =0,...,L-1 \} = K$ implies $\rank{\mtx{B}_z} = K$. The converse is obtained by essentially reversing the line of arguments used to prove this fact in Section \ref{sec:redtommv}.  
We have therefore established that $\dim \spanof \{ [\mbf z(t,f)]_p, p  =0,...,L-1 \} = \rank{\mtx{B}_z}$. 
Analogously, by using the fact that $\Bs{\Gamma}$ contains the expansion coefficients of\linebreak $\{s_{k,m}(t,f), \allowbreak (k,m)\in \Gamma\}$ in the basis $\{b_0(t,f),...,\allowbreak b_{K-1}(t,f)\}$, it can be shown that $\rank{\Bs{\Gamma}} = \dim \spanof \{s_{k,m}(t,f), \allowbreak (k,m)\in \Gamma\}$. 
It now follows, by application of Proposition \ref{prop:cond_uniq_mmv}, that $\MMVN$ correctly recovers the support set $\Gamma$ if and only if \eqref{eq:unique_imv} is satisfied. By equivalence of $\MMVN$ and $\IMV$,  $\IMV$ recovers the correct support set, provided that \eqref{eq:unique_imv} is satisfied. Once $\Gamma$ is known, $\mbf s_\Gamma(t,f)$ is obtained by solving \eqref{eq:sysofeq_rest}. 
\end{proof}

\subsection{Efficient algorithms for solving $\MMVN$ \label{sec:effrecalg}}

Solving the MMV problem $\MMVN$~is NP-hard \cite{davis_adaptive_1997}. Various alternative approaches with different performance-complexity tradeoffs are available in the literature. 
MMV-variants of standard algorithms used in single measurement sparse signal recovery, such as orthogonal matching pursuit (OMP) and $\ell_1$-minimization (basis-pursuit)  can be found in \cite{jie_chen_theoretical_2006,cotter_sparse_2005,tropp_algorithms_2006-1}. 
However, the performance guarantees available in the literature for these algorithms fall short of allowing to choose $|\Gamma|$ to be linear in $L$ as is the case in the threshold \eqref{eq:unique_imv}. 
A low-complexity algorithm that provably yields exact recovery under the threshold in \eqref{eq:unique_imv} is based on ideas developed in the context of subspace-based direction-of-arrival estimation, specifically on the MUSIC-algorithm \cite{schmidt_multiple_1986}. It was first recognized in the context of spectrum-blind sampling \cite{feng_universal_1997,bresler_spectrum-blind_2008} that a MUSIC-like algorithm can be used to solve a problem of the form $\IMV$. The algorithm described in \cite{feng_universal_1997,bresler_spectrum-blind_2008} implicitly first reduces the underlying infinite measurement vector problem to a (finite) MMV problem. 
Recently, a MUSIC-like algorithm and variants thereof were proposed \cite{lee_subspace_2012} to solve the MMV problem $\MMVN$  directly. 
As we will see below, this class of algorithms imposes conditions on $(\Gamma, \Bs{\Gamma})$ and will hence not guarantee recovery for all $(\Gamma, \Bs{\Gamma})$. We will present a (minor) variation of the MUSIC algorithm as put forward in \cite{schmidt_multiple_1986}, and used in the context of spectrum blind sampling \cite[Alg. 1]{bresler_spectrum-blind_2008}, in Section \ref{sec:almostall} below.
 
\section{Identification for almost all $H \in \opclassbr$ \label{sec:almostall}}
For $K > 1$, Theorem \ref{thm:uniqueness_IMV} hints at a potentially significant improvement over the worst-case threshold underlying Theorem \ref{th:almost_all} whose proof will be presented next. 
The basic idea of the proof is to show that $\IMV$ applied to $\mbf z(t,f) = \mbf A_{\Gamma}  \mbf s_\Gamma(t,f), \, (t,f) \in \Un$, recovers the correct solution if the set
\begin{equation}
\text{$\{s_{k,m}(t,f), (k,m) \in \Gamma\}$ is linearly independent on $\Un$}.
\label{eq:condlinind}
\end{equation}
\begin{proof}[Proof of Theorem \ref{th:almost_all}]
Condition \eqref{eq:condlinind} implies that $\dim \spanof \{s_{k,m}(t,f), (k,m)\in \Gamma\} = |\Gamma|$. Therefore, with $K=|\Gamma|$ in Theorem \ref{thm:uniqueness_IMV}, we get that $\IMV$ delivers the correct solution if
 $|\Gamma|<L$, i.e., if $|\Gamma|/L = \area(M_\Gamma) < 1$, which is guaranteed by $\area(M_\Gamma) \leq \SSF <1$. 
\end{proof}

We next present an algorithm that provably recovers $H \in \opsetaa$ with $\SSF <1$, i.e., almost all $H \in \opclassbr$ with $\SSF <1$. Specifically, this low-complexity MUSIC-like algorithm solves\footnote{Note that this does not contradict the fact that $\IMV$ is NP-hard (as noted before), since it ``only'' solves $\IMV$ for \emph{almost all} $\vect{s}(t,f)$. 
} $\MMVN$ (which is equivalent to $\IMV$) and can be shown to identify the support set $\Gamma$ correctly for $\SSF < 1$ provided that Condition \eqref{eq:condlinind} is satisfied. The algorithm is a minor variation of the MUSIC algorithm as put forward in \cite{schmidt_multiple_1986}, and used in the context of spectrum blind sampling \cite[Alg. 1]{bresler_spectrum-blind_2008}.

\begin{theorem}
The following algorithm recovers all $H \in \opsetaa$, provided that $\SSF <1$.
\begin{enumerate}
\setlength\itemindent{1cm}
\item[Step 1)] Given the measurement $\vect{z}(t,f)$, find a basis (not necessarily orthogonal) $\{b_0(t,f),...,\allowbreak b_{K-1}(t,f)\}$, for the space spanned by $\{[\mbf z(t,f)]_p, p=0,...,L-1\}$, where $K \defeq \dim \allowbreak\spanof \{[\mbf z(t,f)]_p, p=0,...,L-1 \}$, and determine the coefficient matrix $\mtx{B}_{z}$ in the expansion $\mbf z(t,f) =  \mtx{B}_{z} \vect{b}(t,f)$. 
\item[Step 2)]  Compute the matrix $\mbf U_n$ of eigenvectors of $\mtx{Z} \defeq \mtx{B}_{z} \herm{\mtx{B}}_{z}$ corresponding to the zero eigenvalues of $\vect{Z}$. 
\item[Step 3)] 
Identify $\Gamma$ with the indices corresponding to the columns of $\herm{\mbf U}_n \mbf A_{\mbf c}$ that are equal to $\mbf 0$. 
\end{enumerate}
\label{thm:music_recovery}
\end{theorem} 

\begin{remark*}
In the remainder of the paper, we will refer to Steps $2)$ and $3)$ above as the MMV-MUSIC algorithm. As shown next, the MMV-MUSIC algorithm provably solves the MMV problem $\MMVN$ given that $\Bs{\Gamma}$ has full rank $|\Gamma|$. 
\end{remark*}

 \begin{proof}[Proof of Theorem \ref{thm:music_recovery}]
 The proof is effected by establishing that for $\SSF<1$ under Condition \eqref{eq:condlinind} 
 the support set $\Gamma$ is uniquely specified through the indices of the columns of $\herm{\mbf U}_n \mbf A_{\mbf c}$ that are equal to $\mbf 0$. 
To see this, we first obtain from \eqref{eq:sysofeq_incoeff} (where $\mtx{A}_{\Gamma}$ and $\mtx{B}_{\Gamma}$ are as defined in Section \ref{sec:reconstruct})
\begin{equation}
 \mtx{Z} =\mtx{B}_{z} \herm{\mtx{B}}_{z}= \mbf A_{\Gamma} \underbrace{\Bs{\Gamma} \herm{\mtx{B}}_{\Gamma}}_{\mtx{S}_\Gamma} \herm{\mtx{A}}_{\Gamma}.
\label{eq:mmvZSrel}
\end{equation}
Next, we perform an eigenvalue decomposition of $\mtx{Z}$ in \eqref{eq:mmvZSrel} to get
\begin{equation}
 \mtx{Z} = 
\begin{bmatrix}
\mbf U_z  & \mbf U_n 
\end{bmatrix}
  \begin{bmatrix} \mbf \Lambda_z & \mbf 0\\   \mbf 0 &  \mbf 0   \end{bmatrix}  
  \begin{bmatrix}
  \herm{\mbf U}_z \\
  \herm{\mbf U}_n
  \end{bmatrix}
  = \mbf U_z \mtx \Lambda_z \herm{\mbf U}_z
  =  \mbf A_{\Gamma} \mtx{S}_\Gamma \herm{\mtx{A}}_{\Gamma}
\label{eq:ev2}
\end{equation}
where $\mbf U_{z}$ contains the eigenvectors of $\mbf Z$ corresponding to the non-zero eigenvalues of $\mbf Z$. 
As mentiond in Section \ref{sec:suff}, each set of $L$ or fewer columns of $\mbf A_{\mbf c}$ is necessarily linearly independent, if $\vect{c}$ is chosen judiciously. Hence $\mbf A_\Gamma$ has full rank if $|\Gamma| \leq L$, which is guaranteed by $\SSF = |\Gamma|/L<1$. 
Thanks to Condition \eqref{eq:condlinind}, $\dim \spanof \{s_{k,m}(t,f), (k,m)\in \Gamma\} =  |\Gamma|$ and hence $\rank{\Bs{\Gamma}} = |\Gamma|$ (this was shown in the proof of Theorem \ref{thm:uniqueness_IMV}), which due to $\mtx{S}_\Gamma=\Bs{\Gamma} \herm{\mtx{B}}_{\Gamma}$ implies that $\rank{\mtx{S}_\Gamma} =|\Gamma|$. 
 Consequently, we have 
 \begin{equation}
 \range(\mbf A_{\Gamma}) = \range(\mbf A_{\Gamma}  \mbf S_\Gamma \herm{\mbf A}_{\Gamma})= \range(\mbf U_z \mbf \Lambda_z \herm{\mbf U}_{z}) = \range(\mbf U_z)
\label{eq:rankeq}
\end{equation}
where the second equality follows from \eqref{eq:ev2}. 
 $\range(\mbf U_n)$ is the orthogonal complement of $\range(\mbf U_z)$ in ${\mathbb C}^L$. It therefore follows from \eqref{eq:rankeq} that $\herm{\mbf U}_n \mbf A_\Gamma=\mbf 0$. 
Hence, the columns of $\herm{\mbf U}_n \mbf A_{\mbf c}$ that correspond to indices $(k,m) \in \Gamma$ are equal to $\mbf 0$. 

It remains to show that no other column of $\herm{\mbf U}_n \mbf A_{\mbf c}$ is equal to $\mbf 0$. This will be accomplished through proof by contradiction. 
Suppose that $\herm{\mbf U}_n \mbf a = \mbf 0$ where $\mbf a$  is any column of $\mbf A_{\mbf c}$ corresponding to an index pair $(k',m') \notin \Gamma$. 
Since  $\range(\mbf U_n)$ is the orthogonal complement of $\range(\mbf U_z)$ in ${\mathbb C}^L$, $\mbf a \in \range(\mbf U_z)= \range(\mbf A_{\Gamma})$. 
This would, however, mean that the $L$ or fewer columns of $\mbf A_{\mbf c}$ corresponding to the indices $(k,m) \in \{\Gamma \cup (k',m')\}$ would be linearly dependent, which stands in contradiction to the fact that each set of $L$ or fewer columns of $\mbf A_{\mbf c}$ must be linearly independent. 
\end{proof}

\section{Compressive system identification and discretization \label{sec:discior}}

\newcommand{\Tint}{V}
\newcommand{\Bint}{B}
\newcommand{\btlim}[1]{#1} 
\newcommand{\Dtau}{E} 	
\newcommand{\Dnu}{D}	
\newcommand{\Ud}{U_d}
\newcommand{\MMVD}[0]{(\text{P0$^*$})}	
\newcommand{\MMVE}[0]{(\text{P0$^\star$})}	

 \newcommand{\dsfapprox}[2]{ \hat{\sfunc}[#1,#2 ]}
 \newcommand{\dsfapproxs}[0]{ \hat{\sfunc}}

The results presented thus far rely on probing signals of infinite bandwidth and infinite duration. It is therefore sensible to ask whether identification under a bandwidth-constraint on the probing signal and under limited observation time of the operator's corresponding response is possible. We shall see that the answer is in the affirmative with the important qualifier of identification being possible up to a certain resolution limit (dictated by the time- and bandwidth constraints). 
The discretization through time- and band-limitation underlying the results in this section will involve approximations that are, however, not conceptual. 

The discussion in this section serves further purposes. First, it will show how the setups in \cite{bajwa_learning_2008,bajwa_identification_2011,herman_high-resolution_2009,pfander_identification_2008} can be obtained from ours through discretization induced by band-limiting the input and time-limiting and sampling the output signal. More importantly, we find that, depending on the resolution induced by the discretization, the resulting recovery problem can be an MMV problem. 
The recovery problem in \cite{bajwa_learning_2008,bajwa_identification_2011,herman_high-resolution_2009,pfander_identification_2008} is a standard (i.e., single measurement) recovery problem, but multiple measurements improve the recovery performance significantly, according to the recovery threshold in Theorem \ref{thm:uniqueness_IMV}, and are crucial to realize recovery beyond $\SSF = 1/2$. 
Second, we consider the case where the support area of the spreading function is (possibly significantly) below the identification threshold $\SSF \leq 1/2$, and we show that this property can be exploited to identify the system while undersampling the response to the probing signal. In the case of channel identification, this allows to reduce identification time, and in radar systems it leads to increased resolution.

\subsection{\label{dttblim} Discretization through time- and band-limitation}

Consider an operator $H \in \opclassbr$ and an input signal $\btlim{x}(t)$ that is band-limited to $[0,\Bint)$, and perform a time-limitation of the corresponding output signal $y(t) = (Hx)(t)$ to $[0,\Tint)$. Then, the input-output relation \eqref{eq:hilbertschmidt_s} becomes  (for details, see, e.g. \cite{boelcskei_fundamentals_2012})
\begin{equation}
y(t) \!\defeq\! (H\prsig)(t) \!=\! \frac{1}{B\Tint} \sum_{r\in \mb Z}\sum_{l \in \mb Z} \overline{\sfunc}\! \left(\frac{r}{B},\frac{l}{\Tint} \right) \btlim{\prsig} \!\left( t-\frac{r}{B}  \right) e^{j 2 \pi \frac{l}{\Tint} t}
\label{eq:sfunc_discrete}
\end{equation}
for $0\leq t\leq \Tint$, where 
\begin{align}
\overline{\sfunc}(\tau,\nu) \label{eq:btlimiorel} = 
B\Tint \!\! \int_{\nu'} \!\int_{\tau'} \!\!\! \sfunc(\tau',\nu')    \sinc( (\tau \!- \!\tau') B ) \sinc( (\nu \!-\! \nu')  \Tint  )  d\tau' d\nu' . \nonumber 
\end{align}
Band-limiting the input and time-limiting the corresponding output hence leads to a discretization of the input-output relation, with ``resolution'' $1/\Bint$ in $\tau$-direction and $1/\Tint$ in $\nu$-direction. It follows from \eqref{eq:btlimiorel} that for a compactly supported $\sfunc(\tau,\nu)$ the corresponding quantity $\overline{\sfunc}(\tau,\nu)$ will not be compactly supported. 
Most of the volume of $\overline{\sfunc}(\tau,\nu)$ will, however, be supported on $M_\Gamma +(-1/\Bint,1/\Bint) \times (-1/\Tint,1/\Tint)$, so that we can approximate \eqref{eq:sfunc_discrete} by restricting summation to the indices $(r,l)$ satisfying $(r/B, l/\Tint) \in M_\Gamma$. 
Note that the quality of this approximation depends on the spreading function as well as on the parameters $B,V,T$, and $L$. 
Here, we assume that $1/(TL) \leq 1/\Tint$ and $T \leq 1/\Bint$. These constraints are not restrictive as they simply mean that we have at least one sample per cell. 
We will henceforth say that $H$ is identifiable with resolution $(1/\Bint,1/\Tint)$, if it is possible to recover $\overline{\sfunc} \left(r/\Bint,l/\Tint \right)$, for $(r/B, l/\Tint) \in M_\Gamma$, from $y(t)$. 
We will simply say ``$H$ is identifiable'' whenever the resolution is clear from the context. In the ensuing discussion $\overline{\sfunc} \left(r/B,l/\Tint \right), r,l \in \mathbb Z$, is referred to as the discrete spreading function. 
The maximum number of non-zero coefficients of the discrete spreading function to be identified is $\area(M_\Gamma) \Bint \Tint$. 

Next, assuming that $\nu_{\max}$, as defined in \eqref{eq:suppgridcont}, satisfies $\nu_{\max} \ll B$, it follows that $y(t)$ is approximately band-limited to $[0,\Bint)$. From \cite{slepian_bandwidth_1976} we can therefore conclude that $\btlim{y}(t)$ lives in a $\Bint \Tint$-dimensional signal space (here, and in the following, we assume, for simplicity, that $\Bint \Tint$ is integer-valued), and can hence be represented through $\Bint \Tint$ coefficients in the expansion in an orthonormal basis for this signal space. The corresponding basis functions can be taken to be the prolate spheroidal wave functions \cite{slepian_bandwidth_1976}. 
\newcommand{\A}{\vect{A}}
\newcommand{\s}{\vect{s}}
\!Denoting the vector containing the corresponding expansion coefficients by $\vect{y}$, the input-output relation \eqref{eq:sfunc_discrete} becomes
\begin{equation}
\vect{y} = \A \s
\label{eq:syseqyAs}
\end{equation}
where the columns of $\A\in \complexset^{\Bint \Tint \times \Bint \Tint \tau_{\max} \nu_{\max}}$ contain the expansion coefficients of the time-frequency translates $\prsig \! \left( t-r/B  \right) e^{j 2 \pi \frac{l}{\Tint} t}$ in the prolate spheroidal wave function set, 
and $\s\in \complexset^{\Bint \Tint \tau_{\max} \nu_{\max}}$ contains the samples $\overline{\sfunc} \left(r/B,l/\Tint \right)$ for $(r/\Bint,l/\Tint) \in [0,\tau_{\max})\times [0,\nu_{\max})$, of which at most $\area(M_\Gamma) \Bint \Tint$ are non-zero, with, however, unknown locations in the $(\tau,\nu)$-plane. 
We next show that the recovery threshold $\SSF \leq 1/2$ continues to hold, independently of the choice of $\Bint$ and $\Tint$. 

\paragraph*{Necessity}
It follows from \cite[Cor. 1]{donoho_optimally_2003} that $\norm[0]{\s} \leq (\spark(\A)-1)/2$ is necessary to recover $\s$ from $\vect{y}$ given $\A$. 
 With $\norm[0]{\s} = \area(M_\Gamma) \Bint \Tint$ and $\spark(\A) \leq \min(\Bint \Tint, \Bint \Tint \tau_{\max} \nu_{\max} )+1 \leq  \Bint \Tint+1$, which follows trivially\footnote{Note that for $\A \in \complexset^{m\times n}$, we trivially have $\spark(\A) \leq \min(m,n)+1$.} since $\A$ is of dimension $\Bint \Tint \times \Bint \Tint \tau_{\max} \nu_{\max}$, we get $\area(M_\Gamma) \Bint \Tint \leq \Bint \Tint/2$ and hence $\area(M_\Gamma) \leq 1/2$. Since, by definition, 
 $\area(M_\Gamma) \leq \SSF$ we have shown that $\SSF \leq 1/2$ is necessary for identifiability.

\paragraph*{Sufficiency}
\newcommand{\D}{D}
Sufficiency will be established through explicit construction of a probing signal $x(t)$ and by sampling the corresponding output signal $y(t)$. 
Since $\btlim{y}(t)$ is (approximately) band-limited to $[0,\Bint)$, we can sample $\btlim{y}(t)$ at rate $\Bint$, which results in
\begin{align}
\btlim{y}\left(\frac{n}{\Bint}\right)  = \frac{1}{B\Tint} \sum_{r=0}^{\Bint \tau_{\max} -1} \sum_{l=0}^{\Tint \nu_{\max} -1}  \overline{\sfunc} \left(\frac{r}{B},\frac{l}{\Tint} \right) \btlim{\prsig} \left(\frac{n-r}{B}  \right) e^{j 2 \pi \frac{ln}{\Bint \Tint}} \label{eq:sfunc_discretesample}
\end{align}
for $n = 0,...,\Bint \Tint -1$. 
In the following, denote the number of samples of $\overline{\sfunc}(\tau,\nu)$ per cell $\Un + \left(k \csamp ,  m/(\csamp L)\right), (k,m)\in \Sigma$, in $\tau$-direction as $\Dtau$ and in $\nu$-direction as $\Dnu$; see Figure \ref{fig:disctnp} for an illustration. Note that, since $\Un = \left[ 0, \csamp  \right) \times \left[0, 1/ (\csamp L)  \right)$, and $\overline{\sfunc}(\tau,\nu)$ is sampled at integer multiples of
$1/\Bint$ in $\tau$-direction and of $1/\Tint$ in $\nu$-direction, we have $\Dtau=\Bint \csamp$ and $\Dnu = \Tint/(\csamp L)$. 
  The number of samples per cell $\Dtau \Dnu = \Bint \Tint / L$ will turn out later to equal the number of measurement vectors in the corresponding MMV problem. To have multiple measurements, and hence make identification beyond $\SSF=1/2$ possible, it is therefore necessary that $\Bint \Tint$ is large relative to $L$. 
As mentioned previously, the samples $\overline{\sfunc} \left(r/\Bint,l/\Tint \right)$, for $(r/B, l/\Tint) \in M_\Gamma$, fully specify the discrete spreading function. We can group these samples into the active cells, indexed by $\Gamma$, by assigning $\overline{\sfunc} \left((r +\Dtau k)/\Bint,(l+ \Dnu m)/\Tint \right)$, for $(r,l) \in \Ud$, to the cell with index $(k,l)$, where $(k,l) \in \Gamma$, and $\Ud \defeq   \{0,...,\Dtau-1\} \times \{0,...,\Dnu-1\}$. 

\begin{figure}
\centering
	\includegraphics{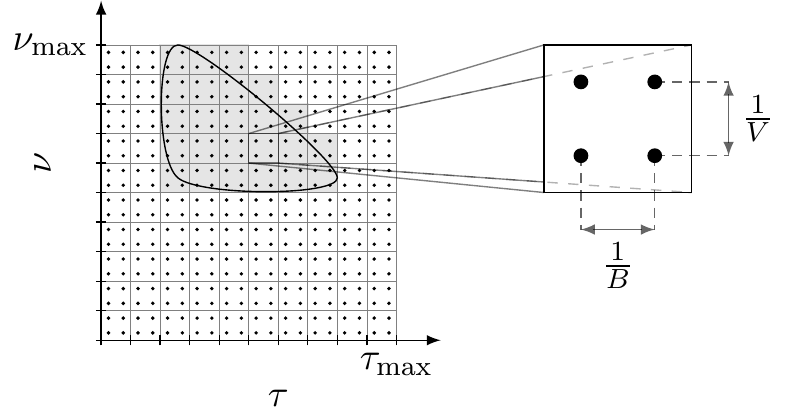}
\caption{\label{fig:disctnp} Discretization of the $(\tau,\nu)$-plane, with $\Dtau = \Dnu=2$.}
\end{figure}
The probing signal $x(t)$ is taken to be such that
\[
\btlim{\prsig}(m/\Bint)
= \begin{cases}
 c_{-k},\;  & \text{for }m = \Dtau k,\; \forall k \in \mathbb Z \\
 0, & \text{otherwise}
\end{cases}
\]
where the coefficients $c_k=c_{k+L}, \forall k \in \mathbb Z$, are chosen as discussed in Section \ref{sec:suff}. Note that the sequence $\btlim{\prsig}(m/\Bint)$ is $\Dtau L$-periodic. 
Algebraic manipulations yield the discrete equivalent of \eqref{eq:sysofeq} as
\begin{equation}
\btlim{\mbf z}[n,r] = \mbf A_{\mbf c} \btlim{\vect{s}}[n,r], \quad  (n,r)  \in \Ud.
\label{eq:syseqdisc}
\end{equation}
Here, $\mbf A_{\mbf c}$ was defined in \eqref{eq:defAc}, and 
\[
\left[\btlim{\mbf z}[n,r] \right]_p \defeq \btlim{z}_p[n,r]  e^{-j2\pi \frac{rp}{\Dnu L} }, \quad p = 0,..., L-1
\]
with
$
\btlim{z}_p[n,r] \defeq Z_{y}^{(\Dtau L,\Dnu )}[n+\Dtau p ,r] 
$,
where  
\begin{equation*}
	Z_{y}^{(\Dtau L,\Dnu)}[n,r]   \defeq  \frac{1}{\Dnu} \sum_{q  =  0}^{\Dnu-1}   y[n  + \Dtau L q]   
	e^{-j 2 \pi  \frac{qr}{\Dnu} }, 
\end{equation*}
for $0 \leq n \leq \Dtau L-1, \; 0\leq r \leq \Dnu-1$, 
 is the discrete Zak transform (with parameter $\Dtau L$) \cite{bolcskei_discrete_1997} of the sequence $y[n]$. For general properties of the discrete Zak transform we refer to \cite{bolcskei_discrete_1997}. 
Further, $\btlim{\mbf s}[n,r]  \defeq [\btlim{s}_{0,0}[n,r],\allowbreak \btlim{s}_{0,1}[n,r], \allowbreak ...,\allowbreak \btlim{s}_{0,L-1}[n,r], \btlim{s}_{1,0}[n,r],\allowbreak..., \btlim{s}_{L-1,L-1}[n,r] ]^T$ with 
 \begin{equation}
 \btlim{s}_{k,m}[n,r] \defeq    \overline{\sfunc} \!\left(\frac{n+\Dtau k}{\Bint },\frac{ r+\Dnu m}{\Tint} \right)   e^{j2 \pi \frac{n (r+\Dnu m)}{\Dtau \Dnu L}} 
 \label{def:skmdiscrete}
\end{equation}
for $\quad  (n,r) \in \Ud, \;(k,m) \in \Sigma$. 
Note that $\btlim{s}_{k,m}[n,r], (n,r)\in \Ud, (k,m) \in \Gamma$, fully specifies the discrete spreading function.

The identification equation \eqref{eq:syseqdisc}  can be rewritten as 
\begin{equation}
\mtx{Z} = \mbf A_{\mbf c} \mtx{S}
\label{eq:discZAS}
\end{equation}
where the columns of $\mtx{Z} \in \mbf \complexset^{L\times \Dtau \Dnu}$ and $\mtx{S} \in \complexset^{ L^2 \times \Dtau \Dnu}$ are given by the vectors $\btlim{\mbf z}[n,r]$ and $\btlim{\vect{s}}[n,r]$, respectively, $(n,r) \in \Ud$. Hence, each row of $\mtx{S}$ corresponds to the samples of $\overline{\sfunc}$ in one of the $L^2$ cells. Since the number of samples per cell, $\Dtau \Dnu$, is equal to the number of columns of $\mtx{S}$, we see that the number of samples per cell corresponds to the number of measurements in the MMV formulation \eqref{eq:discZAS}. 
 Denote the matrix obtained from $\mtx{S}$ by retaining the rows corresponding to the active cells, indexed by $\Gamma$, by $\mtx{S}_\Gamma$ and let $\mbf A_\Gamma$ be the matrix containing the corresponding columns of  $\mbf A_{\mbf c}$. Then \eqref{eq:discZAS} becomes
\begin{equation}
\mtx{Z} = \mbf A_{\Gamma} \mtx{S}_\Gamma. 
\label{eq:syseqdres}
\end{equation}
Once $\Gamma$ is known, \eqref{eq:syseqdres} can be solved for $\mtx{S}_\Gamma$. Hence, recovery of the discrete spreading function amounts to identifying $\Gamma$ from the measurements $\mtx{Z}$, which  
can be accomplished by solving the following MMV problem:
\begin{equation*}
\MMVD \; \begin{cases} 
		\text{minimize} 	& |\Gamma| \\
		\text{subject to} 	&\mtx{Z} = \mbf A_{\Gamma} \mtx{S}_\Gamma
\end{cases}
\end{equation*}
where the minimization is performed over all $\Gamma \subseteq \Sigma$ and all corresponding $\mtx{S}_\Gamma \in \complexset^{|\Gamma|\times \Dtau \Dnu}$. 
 It follows from Proposition \ref{prop:cond_uniq_mmv} that $\Gamma$ is recovered exactly from $\vect{Z}$ by solving $\MMVD$, whenever $|\Gamma| < (L+K)/2$, where $K = \rank{\mtx{S}_\Gamma}$. Correct recovery is hence guaranteed whenever $|\Gamma| \leq L/2$.  
 Since $|\Gamma|/L = \area(M_\Gamma)$ and $\area(M_\Gamma) \leq \SSF$, this shows that $\SSF \leq 1/2$ is sufficient for identifiability. 
 
As noted before, $\MMVD$ is NP-hard. However, if $\mtx{S}_\Gamma \text{ has full rank } |\Gamma|$
then MMV-MUSIC 
 provably recovers $\Gamma$ with $|\Gamma| < L$, i.e., when $\SSF < 1$ (this was shown in the proof of Theorem \ref{thm:music_recovery}). For $\mtx{S}_\Gamma \in \complexset^{|\Gamma|\times \Dtau \Dnu}$ to have full rank $|\Gamma|$, it is necessary that the number of samples satisfy $\Dtau \Dnu \geq |\Gamma|$. For $\Dtau \Dnu \geq |\Gamma|$ almost all $\mtx{S}_\Gamma$ have full rank $|\Gamma|$. 
The development above shows that the MMV aspect of the recovery problem is essential to get recovery for values of $\SSF$ beyond $1/2$. 

We conclude this discussion by noting that the setups in \cite{bajwa_learning_2008,pfander_identification_2008} in the context of channel estimation and in \cite{herman_high-resolution_2009} in the context of compressed sensing radar are structurally equivalent to the discretized operator identification problem considered here, with the important difference of the MMV aspect of the problem not being brought out in \cite{bajwa_learning_2008,herman_high-resolution_2009,pfander_identification_2008}.

\subsection{Compressive identification}

In the preceding sections, we showed 
under which conditions identification of an operator is possible if the operator's spreading function support region is not known prior to identification. 
We now turn to a related problem statement that is closer to the philosophy of sparse signal recovery, where the goal is to reconstruct sparse objects, such as signals or images, by taking fewer measurements than mandated by the object's ``bandwidth''. 
We consider the discrete setup \eqref{eq:sfunc_discretesample} and assume that $\SSF$ is (possibly significantly) smaller than the identifiability threshold $1/2$. Concretely, set $\SSF = P/(2L)$ for an integer $P\leq L$. 
We ask whether this property can be exploited to recover the discrete spreading function from a subset of the samples $\{y(n/\Bint)$, $n = 0,...,\Bint \Tint -1\}$ only. We will see that the answer is in the affirmative, and that the corresponding practical implications are significant, as detailed below.  

For concreteness, we assume that $\supp(\overline{\sfunc}) = M_\Gamma \subseteq M_\Phi$, with $\Phi = \{0,...,\lfloor \sqrt{L} \rfloor -1\} \times \{0,...,\lfloor \sqrt{L} \rfloor -1\}$. To keep the exposition simple, we take $\Dtau \Dnu=1$, in which case $\mtx{S}_\Phi$ becomes a vector. 
Note that, since $\area(M_\Phi) \leq 1$ (this follows from $\area(\Un) \lfloor \sqrt{L} \rfloor^2 = (1/L) \lfloor\sqrt{L} \rfloor^2 \leq 1$), the operator can be identified by simply solving $\vect{Z} = \vect{A}_{\Phi} \vect{S}_\Phi$ for $\vect{S}_\Phi$, which we will refer to as ``reconstructing conventionally''.  
Here $\vect{A}_{\Phi}$ and  $\vect{S}_\Phi$ contain the columns of $\vect{A}_{\vect{c}}$ and rows of $\vect{S}$, respectively, corresponding to the indices in $\Phi$. 

Since $\SSF = P/(2L)$, the area $\area(M_\Gamma)$ of the (unknown) support region $M_\Gamma$ of the spreading function  satisfies  $\area(M_\Gamma) \leq P/(2L)$. We  next show that the discrete spreading function can be reconstructed  from only $P$ of the $L$ rows of $\vect{Z}$. The index set corresponding to these $P$ rows is denoted as $\Omega$, and is an (arbitrary) subset of $\{0,...,L-1\}$ (of cardinality $P$). Let $\vect{Z}^\Omega$ and $\A_{\vect{c}}^\Omega$ be the matrices corresponding to the rows of $\vect{Z}$ and $\A_{\vect{c}}$, respectively, indexed by $\Omega$. The matrix $\vect{Z}^\Omega$ is a function of the samples $\{\btlim{y}\left( n/\Bint \right)\colon n \in \Omega\}$ only; hence, reconstruction from $\vect{Z}^\Omega$ amounts to reconstruction from an undersampled version of $y(t)$. 
Note that we cannot reconstruct the discrete spreading function by simply inverting $\A_{\Phi}^\Omega \in \complexset^{P \times \lfloor \sqrt{L} \rfloor^2}$ since $\A_{\Phi}^\Omega$ is a wide matrix\footnote{The special case $L \geq P\geq  \lfloor \sqrt{L} \rfloor^2$ is of limited interest and will not be considered.}. 
Next, \eqref{eq:discZAS} implies (see also \eqref{eq:syseqdres}) that
\[
\vect{Z}^\Omega = \mbf A_{\mbf c}^\Omega \mtx{S} = \mbf A_{\Gamma}^\Omega \mtx{S}_\Gamma.
\]
Theorem 4 in \cite{lawrence_linear_2005} establishes that for almost all $\mbf c$, $\spark( \vect{A}_{\mbf c}^\Omega ) = P$. Hence, according to Proposition \ref{prop:cond_uniq_mmv}, $\mtx{S}_\Gamma$ can be recovered uniquely from $\vect{Z}^\Omega$ provided that $|\Gamma| \leq P/2$ and hence $\SSF \leq P/(2L)$, by solving 
\begin{equation*}
\MMVE \; \begin{cases} 
		\text{minimize} 	& |\Gamma| \\
		\text{subject to} 	& \vect{Z}^\Omega =  \mbf A_{\Gamma}^\Omega \mtx{S}_\Gamma
\end{cases}
\end{equation*}
where the minimization is performed over all $\Gamma \subseteq \Sigma$ and all corresponding $\mtx{S}_\Gamma \in \complexset^{|\Gamma|}$. 

We have shown that identification from an undersampled observation $y(t)$ is possible, and the undersampling factor can be as large as $P/L$. A similar observation has been made in the context of radar imaging \cite{baraniuk_compressive_2007}. 
Recovery of $\mtx{S}_\Gamma$ from $\vect{Z}^\Omega$ has applications in at least two different areas, namely in radar imaging and in channel identification. 

\paragraph*{Increasing the resolution in radar imaging}
In radar imaging, targets correspond to point scatterers with small dispersion in the delay-Doppler plane. Since the number of targets is typically small, the corresponding spreading function is sparsely supported \cite{herman_high-resolution_2009}. 
In our model, this corresponds to a small number of the $\overline{\sfunc} \left(\frac{r}{B},\frac{l}{\Tint} \right)$ in  \eqref{eq:sfunc_discretesample} being non-zero. 
Take $\Omega = \{0,...,P-1\}$. The discussion above then shows that, since only the samples $\{y(n/\Bint), n \in \Omega\}$, which in turn only depend on $y(t)$ for  $t\in [0,\Tint P/L]$, are needed for identification, it is possible to identify the discrete spreading function from the ``effective'' observation interval $[0,\Tint  P/L)$, while keeping the resolution in $\nu$-direction at $1/\Tint$. 
If we were to reconstruct conventionally, given only the observation of $y(t)$ over the interval $[0,\Tint P/L)$, the induced resolution in $\nu$-direction (see Figure \ref{fig:disctnp}) would only be $L/(P\Tint)$. 

\paragraph*{Saving degrees of freedom in channel identification}
Next, consider the problem of channel identification, and take again $\Omega = \{0,...,P-1\}$. 
As discussed before, $\vect{Z}^\Omega$ is a function of the samples $\{\btlim{y}\left( n/\Bint \right)\colon n \in \Omega\}$ only, which, by careful inspection of \eqref{eq:sfunc_discretesample}, are seen to depend only on $\{\btlim{x}(n/\Bint), \allowbreak n = -(\lfloor \sqrt{L} \rfloor -1),..., P-1\}$. 
 We can therefore conclude that it suffices to observe $y(t)$ over the interval $[0,\Tint  P/L)$. 
Conceptually, this means that the time needed to identify (learn) the channel is reduced, which leaves, e.g., more degrees of freedom to communicate over the channel.

\section{Numerical results \label{sec:numerical_recres}}

We present numerical results quantifying the impact of additive noise and of the choice of $\vect{c}$ on the performance of different identification algorithms. Specifically, we consider the discrete setting\footnote{We consider the discrete setting as any numerical simulation of the continuous setting will involve a discretization.} \eqref{eq:sfunc_discretesample} and evaluate two probing sequences. The first one is obtained by sampling i.i.d. uniformly from the complex unit disc, the resulting sequence is denoted by $\vect{c}_r$. Since for almost all $\vect{c}_r$, each $L\times L$ submatrix of $\mtx{A}_{\vect{c}}$ has full rank for $L$ prime \cite[Thm. 4]{lawrence_linear_2005}, $\vect{c}_r$ will allow recovery for all operators with $\SSF \leq 1/2$ and for almost all operators with $\SSF < 1$, in both cases, with probability one, with respect to the choice of $\vect{c}_r$. 
The second probing sequence is the Alltop sequence\footnote{
The Alltop sequence was also used in \cite{herman_high-resolution_2009} as probing sequence, motivated by the fact that its mutual coherence attains the Welch lower bound (for $L$ prime). 
}
\cite{alltop_complex_1980}, denoted by $\vect{c}_a$, and defined as
\[
[\vect{c}_{a}]_i = \frac{1}{\sqrt{L}} e^{\frac{j2\pi }{L} i^3}, \quad i = 0,...,L-1. 
\]
We compare two different algorithms for solving the MMV problem $\MMVD$, namely the MMV-orthogonal matching pursuit (MMV-OMP) algorithm as proposed in \cite{jie_chen_theoretical_2006}, and MMV-MUSIC\footnote{
In the noisy case, MMV-MUSIC identifies the columns with $\ell_2$-norm smaller than a certain threshold, which in turn depends on the noise level.} as introduced in Section \ref{sec:almostall}. 
We generate the samples $\overline{\sfunc} \left(\frac{r}{B},\frac{l}{\Tint} \right)$ at random, as follows. 
We choose\footnote{
The reason for choosing $L=19$ is that we want $L$ to be prime, as by \cite{lawrence_linear_2005} this guarantees that for almost all $\vect{c}$, each $L\times L$ submatrix of $\mtx{A}_{\vect{c}}$ has full rank.
}  
$L=19$, and vary the support set size $\SSF = |\Gamma|/L$ and the number of samples per cell, $\Dtau \Dnu$. 
For fixed $\SSF = |\Gamma|/L$, and hence fixed $|\Gamma|$, we draw $\Gamma \subseteq \Sigma$ uniformly at random from the set of all support sets with cardinality $|\Gamma|$, 
and assign i.i.d. $\jpg(0,1)$ values to each of the $\Dtau \Dnu$ samples in each of the corresponding cells. 

To analyze the impact of noise, we contaminate the measurement (i.e., $y(n/\Bint)$ in \eqref{eq:sfunc_discretesample}) by i.i.d. Gaussian noise. Recovery performance in the noisy case is quantified through the empirical relative squared error in the discrete spreading function, abbreviated as ERE, which is the empirical expectation of the relative squared error. 
In the noiseless setting, recovery success is declared if the relative squared error in the spreading function is less than or equal to $10^{-5}$. Recovery probabilities and the ERE were obtained from 1000 realizations of $\Gamma$.

\setcounter{paragraph}{0}
\paragraph{Impact of probing sequence}
The results for the noiseless case, depicted in Figure \ref{fig:emprec}, show that the probing sequences $\vect{c}_a$ and $\vect{c}_r$ perform almost equally well. 
We can see, as predicted by Theorem \ref{thm:music_recovery}, that MMV-MUSIC succeeds for $\SSF<1$, provided that $\Dtau \Dnu/L \geq |\Gamma|/L = \SSF$. Specifically, as shown in the proof of Theorem \ref{thm:music_recovery}, MMV-MUSIC succeeds if $\mtx{S}_\Gamma$ has full rank, which is the case with probability one if $\Dtau \Dnu/L \geq |\Gamma|/L = \SSF$, as the entries of $\mtx{S}_\Gamma$ are  i.i.d.  $\jpg(0,1)$. For $\Dtau \Dnu < |\Gamma|$, MMV-MUSIC fails. The performance of MMV-OMP improves in $\Dtau \Dnu$; however, the improvement stagnates at about $\Delta \approx 1/2$. For $\Dtau \Dnu=1$, MMV-OMP outperforms MMV-MUSIC, for all other values of $\Dtau \Dnu$ considered MMV-MUSIC outperforms MMV-OMP. 
\begin{figure}
\centering
\includegraphics{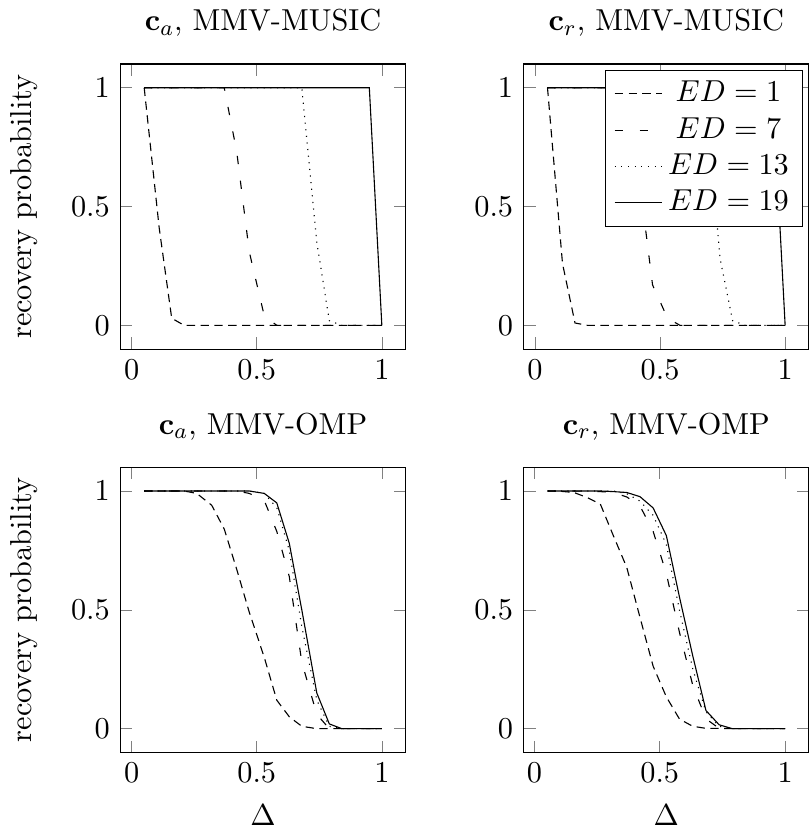}
\caption{\label{fig:emprec} Recovery probabilities for the Alltop sequence and for a randomly generated sequence, for the MMV-MUSIC algorithm in Section \ref{sec:almostall} and the MMV-OMP \cite{jie_chen_theoretical_2006} algorithm. 
}
\end{figure}

\paragraph{Impact of noise}
The results depicted in Figure \ref{fig:noisy} show that the identification process exhibits noise robustness up to $\SSF \approx 1$. When $\Dtau \Dnu/L \geq |\Gamma|/L = \SSF$, the error in recovering the spreading function is small for both identification algorithms, but MMV-MUSIC outperforms MMV-OMP significantly. 
The results in Figure \ref{fig:noisy_2} quantify the noise sensitivity of MMV-MUSIC and MMV-OMP. 


\begin{figure}
\centering
\includegraphics{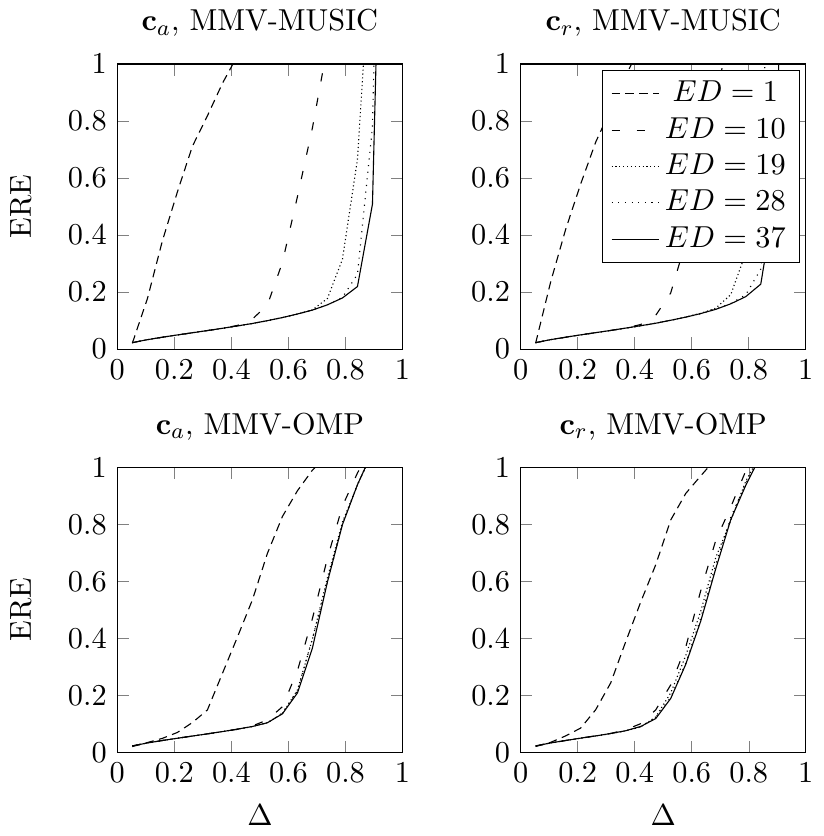}
\caption{\label{fig:noisy} ERE for the Alltop sequence and for a randomly chosen sequence, obtained by MMV-MUSIC and MMV-OMP \cite{jie_chen_theoretical_2006} at $\text{SNR} = 20\text{dB}$.}
\end{figure}

\begin{figure}[]
\begin{center}
\includegraphics{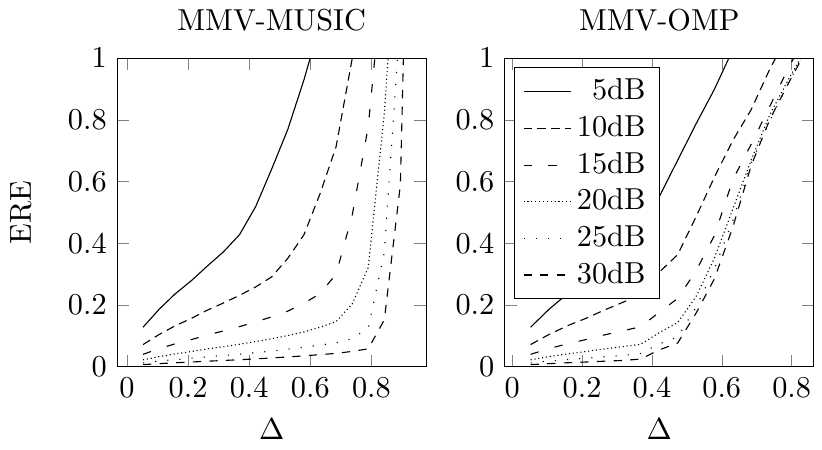}
\end{center}
\caption{\label{fig:noisy_2} 
ERE for a randomly chosen sequence, and $\Dtau \Dnu=19$, obtained by MMV-MUSIC and  MMV-OMP \cite{jie_chen_theoretical_2006} for different SNR values.}
\end{figure}

\section*{Acknowledgment}
The authors would like to thank G. Pfander and V. Morgenshtern for interesting discussions. 

\appendices

\section{Bounded inverse of $T$ \label{app:binv}}

\begin{theorem}
The inverse 
\begin{align}
T^{-1} \colon R_T \to \opclassgen
\label{eq:definv}
\end{align}
of the linear operator
\[
T\colon \opclassgen \to \opto
\]
where $R_T$ is the range of $T$, exists and is bounded if and only if $T$ is bounded below, in the following sense: There exists an $\alpha>0$ such that for all $H_1,H_2 \in \opclassgen$, 
\begin{equation}
\alpha \norm[\opclass]{H_1 -H_2} \leq \norm[]{ T H_1 - T H_2  }.
\label{eq:boundedinversecond}
\end{equation}
\end{theorem}
\begin{proof}
The proof corresponding to the case where $\opclassgen$ satisfies $(H_1 - H_2) \in \opclassgen$ for all $H_1,H_2 \in \opclassgen$ is standard, see e.g. \cite{naylor_linear_2000}. For $(H_1-H_2) \notin \opclassgen$, the proof follows the same steps with minor modifications. 
We first show that \eqref{eq:boundedinversecond} implies bounded invertibility of $T$. If $T H_1 = T H_2$, then from \eqref{eq:boundedinversecond}
\[
\alpha \norm[\opclass]{H_1 -H_2} \leq \norm[]{0}
\] 
and hence necessarily $H_1 = H_2$, which shows that $T$ is injective. Since according to \eqref{eq:definv}, the domain of the inverse is $R_T$, $T$ is also surjective, and hence $T$ is invertible. To show boundedness of $T^{-1}$, set $H_1 = T^{-1} y_1$ and $H_2 = T^{-1} y_2$ for $y_1,y_2 \in R_T$. Using \eqref{eq:boundedinversecond}, we get
\begin{align*}
\norm[]{  y_1 -  y_2}   
= \norm[]{ T  T^{-1} y_1 - T T^{-1} y_2} = \norm[]{ T H_1 - T H_2 } 
\geq \alpha \norm[\opclass]{H_1 -H_2} = \alpha \norm[\opclass]{T^{-1} y_1 - T^{-1} y_2}
\end{align*}
which is
\[
\norm[\opclass]{T^{-1} y_1 - T^{-1} y_2} \leq \frac{1}{\alpha}  \norm[]{  y_1 -  y_2} 
\]
and hence shows that $T^{-1}$ is bounded. 

We next show that bounded invertibility of $T$ implies \eqref{eq:boundedinversecond}. 
Since $T^{-1}$ exists and is bounded, we have, for $\alpha>0$, 
\begin{align*}
\norm[\opclass]{H_1 -H_2} 
= \norm[\opclass]{T^{-1} y_1 - T^{-1} y_2}  
\leq \frac{1}{\alpha} \norm[]{  y_1 -  y_2}  = \frac{1}{\alpha} \norm[]{ T H_1 - T H_2 }.
\end{align*}
\end{proof}

\section{\label{app:lem2}Proof of Lemma \ref{le:boundedness_classr}}
Starting from \eqref{eq:sysofeq_rest}, we get, for fixed $(t,f)\in \Un$,
\begin{align}
 \inf_{\norm[2]{\mbf v} = 1} \, \norm[2]{\mbf A_{\Gamma} \mbf v }  \; \norm[2]{\mbf s_\Gamma (t,f)}   
 \leq   \norm[2]{\mbf z (t,f)}
\leq    \sup_{\norm[2]{\mbf v} = 1} \norm[2]{\mbf A_{\Gamma} \mbf v} \; \norm[2]{\mbf s_\Gamma (t,f)} .
 \label{eq:upperlower}
\end{align}
 Squaring $\norm[2]{\mbf z (t,f)}$ and integrating over $\Un$ yields 
\begin{align}
\int_{\Un}  \norm[2]{\mbf z (t,f)}^2 d(t,f) 
&= \sum_{p = 0}^{L-1} \int_{\Un}  \left| [\mbf z(t,f)]_p \right|^2  d(t,f)  \nonumber \\
&= \sum_{p = 0}^{L-1} \int_{\Un} (\csamp L)^2 \left| z_p(t,f)  \right|^2  d(t,f)   \label{eq:usedefzv} \\
&= (\csamp L)^2  \int_0^{\csamp L}\!\!\!\! \int_0^{1/(\csamp L)}  |\mc Z_{y} (t,f) |^2 d(t,f) \label{eq:usedefzp} \\
&=  \csamp L \, \norm[]{H \prsig}^2
\label{eq:hzeq}
\end{align}
where we used \eqref{eq:defzvect} and \eqref{eq:defpz} for \eqref{eq:usedefzv} and \eqref{eq:usedefzp}, respectively, and \eqref{eq:hzeq} follows since the Zak transform is an isometry (see \eqref{eq:unitarityzak}). 
Similarly, based on \eqref{def:skm} we get 
\begin{equation}
\int_{\Un}   \norm[2]{\mbf s_\Gamma (t,f)}^2 \, d(t,f)  = \norm[]{\sfunc}^2 = \norm[\opclass]{H}^2
\label{eq:normheqints}
\end{equation}
where the last equality follows from \eqref{eq:defHSnorm}. 
Combining \eqref{eq:normheqints} and \eqref{eq:hzeq}  with \eqref{eq:upperlower} yields
\[
\alpha_\Gamma  \norm[\opclass]{H}  \leq   \, \norm[]{H \prsig} \leq   \beta_\Gamma  \norm[\opclass]{H}
\]
 with 
 \[
 \alpha_\Gamma = \frac{1}{\sqrt{TL}} \inf_{\norm[2]{\mbf v} = 1}\norm[2]{ \mbf A_{\Gamma} \mbf v}, \quad \beta_\Gamma = \frac{1}{\sqrt{TL}} \sup_{\norm[2]{\mbf v} = 1}\norm[2]{ \mbf A_{\Gamma} \mbf v}
 \]
 which concludes the proof.

\section{Proof of Proposition \ref{prop:cond_uniq_mmv} \label{app:neccmmv}}

\newcommand{\sparsity}{M}
To prove necessity of \eqref{eq:cond_uniq_mmv}, we show that one can construct a solution $(\Gamma',\Bs{\Gamma'}) \neq (\Gamma, \Bs{\Gamma})$ to $\MMVN$ applied to $ \mtx{B}_{z} = \mbf A_{\Gamma} \Bs{\Gamma}$ with $|\Gamma'| = |\Gamma| \geq (L+K)/2$. 
For any set $\Phi$ of column indices of $\mtx{A}_{\vect{c}}$, with cardinality $|\Phi| = L+K$,  we have that 
$\mtx{A}_\Phi $ has full rank $L$, as each set of $L$ columns of $\mtx{A}_{\vect{c}} \in \complexset^{L\times L^2}$ is linearly independent (as discussed previously, according to \cite[Thm. 4]{lawrence_linear_2005} this holds for almost all $\vect{c}$, and we assume that $\vect{c}$ is chosen accordingly), and hence $\dim \ker (\mtx{A}_\Phi ) = K$. We can therefore conclude that there exists a matrix $\Bs{\Phi} \in \complexset^{(L+K)\times K}$ with $\rank{\Bs{\Phi}} = K$ such that 
\begin{equation}
\mtx{A}_\Phi \Bs{\Phi} = \mtx{0}. 
\label{eq:absphi0}
\end{equation}
We next construct index sets $\Gamma,\Gamma'$ with $\Gamma \cup \Gamma' = \Phi$ and $|\Gamma| = |\Gamma'| = (K+L)/2$. 
Since $\rank{\Bs{\Phi}} = K$, there exists a set of $K$ linearly independent rows of $\Bs{\Phi}$. Let $\Gamma'$ be the index set corresponding to these rows augmented by the indices corresponding to  $(K+L)/2  -K$ arbitrary rows of $\Bs{\Phi}$, and set $\Gamma= \Phi \setminus  \Gamma'$.  
By construction, the matrix formed by the rows indexed by $\Gamma'$, $\Bs{\Gamma'}$, satisfies $\rank{\Bs{\Gamma'}}=K$.  
From \eqref{eq:absphi0}, with $\Bs{\Gamma}$ defined through $\Bs{\Phi} = \transp{[\transp{\mtx{B}}_{\Gamma'}, -\transp{\mtx{B}}_{\Gamma}  ]}$, we have
\begin{equation}
\begin{bmatrix}
\mtx{A}_{\Gamma'}  & \mtx{A}_{ \Gamma}
\end{bmatrix}
\begin{bmatrix}
 \Bs{\Gamma'} \\
 -\Bs{\Gamma}  
\end{bmatrix}
 = \mtx{0}  \iff
 \mtx{A}_{\Gamma'}    \Bs{\Gamma'} =    \mtx{A}_{\Gamma}  \Bs{\Gamma}. 
\label{eq:ABeqAB2}
\end{equation}
It therefore follows from \eqref{eq:ABeqAB2} that ($\Gamma',\Bs{\Gamma'}$) is consistent with $\mtx{B}_{z} =  \mbf A_{\Gamma} \Bs{\Gamma} = \mbf A_{\Gamma'} \Bs{\Gamma'}$, which concludes the proof. 





\end{document}